\documentclass[10pt]{amsart}

\usepackage{amssymb}
\usepackage{amsthm}
\usepackage{amsmath}
\usepackage{aliascnt}
\usepackage{biblatex}
\usepackage{enumerate}
\usepackage[a4paper,twoside=false,hmargin=2cm]{geometry}
\usepackage{hyperref}
\usepackage{microtype}

\let\oldautoref\autoref
\makeatletter
\renewcommand\autoref[1]{\@first@ref#1,@}
\def\@throw@dot#1.#2@{#1}
\def\@set@refname#1{
    \edef\@tmp{\getrefbykeydefault{#1}{anchor}{}}%
    \def\@refname{\@nameuse{\expandafter\@throw@dot\@tmp.@autorefname}s}%
}
\def\@first@ref#1,#2{%
  \ifx#2@\oldautoref{#1}\let\@secondref\@gobble
  \else%
    \@set@refname{#1}
    \@refname~\ref{#1}
    \let\@secondref\@second@ref
  \fi%
  \@secondref#2%
}
\def\@second@ref#1,#2{%
  \ifx#2@ and~\ref{#1}\let\@nextref\@gobble
  \else, \ref{#1}
    \let\@nextref\@next@ref
  \fi%
  \@nextref#2%
}
\def\@next@ref#1,#2{%
   \ifx#2@, and~\ref{#1}\let\@nextref\@gobble
   \else, \ref{#1}
   \fi%
   \@nextref#2%
}
\makeatother

\numberwithin{equation}{section}
\let\oldtheequation\theequation
\makeatletter
\def\tagform@#1{\maketag@@@{\ignorespaces#1\unskip\@@italiccorr}}
\renewcommand{\theequation}{(\oldtheequation)}
\makeatother 

\theoremstyle{plain}

\newtheorem{theorem}{Theorem}[section]

\newaliascnt{lemma}{theorem}
\newtheorem{lemma}[lemma]{Lemma}
\aliascntresetthe{lemma}

\newaliascnt{corollary}{theorem}
\newtheorem{corollary}[corollary]{Corollary}
\aliascntresetthe{corollary}

\newaliascnt{remark}{theorem}
\newtheorem{remark}[remark]{Remark}
\aliascntresetthe{remark} 

\newtheorem*{remark*}{Remark}

\newaliascnt{assumption}{theorem}
\newtheorem{assumption}[assumption]{Assumption}
\aliascntresetthe{assumption}

\theoremstyle{definition}

\newaliascnt{definition}{theorem}
\newtheorem{definition}[definition]{Definition}
\aliascntresetthe{definition}

\newaliascnt{example}{theorem}

\aliascntresetthe{example}

\newaliascnt{algorithm}{theorem}

\aliascntresetthe{algorithm}

\numberwithin{figure}{section}
\numberwithin{table}{section}

\begin{document}

\title[Affine representations of fractional processes]{Affine representations of fractional processes with applications in mathematical finance}
\thanks{Supported in part by SNF grant 149879 and the ETH Foundation. We are grateful to Josef Teichmann, Mario W\"uthrich, and Christa Cuchiero for helpful comments and suggestions.}

\author{Philipp Harms}
\email{philipp.harms@stochastik.uni-freiburg.de}
\address{Department of Mathematical Stochastics, University of Freiburg}
\author{David Stefanovits}
\email{david.stefanovits@gmail.com}
\address{ETH Zurich, Switzerland}

\begin{abstract}
Fractional processes have gained popularity in financial modeling due to the dependence structure of their increments and the roughness of their sample paths. The non-Markovianity of these processes gives, however, rise to conceptual and practical difficulties in computation and calibration. To address these issues, we show that a certain class of fractional processes can be represented as linear functionals of an infinite dimensional affine process. This can be derived from integral representations similar to those of Carmona, Coutin, Montseny, and Muravlev. We demonstrate by means of several examples that this allows one to construct tractable financial models with fractional features.
\end{abstract}



\maketitle

\section{Introduction}

Empirical evidence suggests that certain financial time series may not be captured well by low-dimensional Markovian models. In particular, this applies to short-term interest rates, which tend to have long-range dependence \cite{backus1993emp}, and to volatilities of stock prices, which have rough sample paths and are well described by fractional Brownian motion with small Hurst index \cite{gatheral2014vola}. Dependent increments and rough sample paths are, however, characteristic features of fractional processes.

In this paper we show that certain fractional processes, including fractional Brownian motion and some related processes (see \autoref{rem:fbm:generalizations}), admit representations as linear functionals of an infinite-dimensional affine process. The key idea, which goes back to Carmona, Coutin, Montseny, and Muravlev \cite{carmona1998fractional, carmona2000approximation, muravlev2011representation}, is to express the fractional integral in the Mandelbrot-Van Ness representation of fractional Brownian motion by a Laplace transform: for each $H<1/2$, by the stochastic Fubini theorem,
\begin{equation*}
\int_0^t (t-s)^{H-\frac12}dW_s 
\propto \int_0^t \int_0^\infty e^{-x(t-s)}\frac{dx}{x^{H+\frac12}}dW_s
= \int_0^\infty \int_0^t e^{-x(t-s)}dW_s \frac{dx}{x^{H+\frac12}}.
\end{equation*}
As noted in \cite{carmona1998fractional, carmona2000approximation, muravlev2011representation}, the right-hand side is a superposition of infinitely many Ornstein--Uhlenbeck (OU) processes with varying speed of mean reversion. 

Our contribution is two-fold. First, we introduce the idea of studying fractional processes from an affine point of view. Specifically, we show that the collection of OU processes is an affine process on a state space of $L^1$ or $L^2$ functions. The case $H>1/2$ required a new integral representation of the function $(t-s)^{H-1/2}$ (see \autoref{rem:fbm:lit}). Second, we formulate several financial models with fractional features in the language of affine processes. Specifically, we construct a fractional short rate model where, in contrast to \cite{ohashi2009ftsm} and \cite{biagini2013ftsm}, discounted zero-coupon bond prices are martingales. We also build a fractional version of the stochastic volatility model by Stein and Stein \cite{stein1991}. 

Our result is relevant in mathematical finance and probability for the following reasons. 
First, affine representations of fractional processes provide a natural way of generalizing well-known affine models from semimartingale to fractional settings. 
This can be useful to model quantities which are not restricted by no-arbitrage theory to be semimartingales (e.g.\@ volatilities of stock prices \cite{gatheral2014vola}).
In the present Gaussian setting the full power of the affine machinery admittedly does not come into play because conditional expectations can also be calculated directly from integral representations.
This is, however, not the case for non-Gaussian fractional processes as in \cite{jaber2017affine}.

Second, there is recently a high interest in non-affine fractional volatility models such as the fractional Bergomi and SABR models \cite{mendes2008data,gatheral2014vola}. It is a major challenge to derive short-time, large-time, and wing asymptotics for these models, as well as to develop numerical schemes for pricing and calibration. Hopefully, the Markovian point of view will be helpful for achieving these goals.

Third, the Markovian structure is useful for characterizing the behavior of fractional Brownian motion after a stopping time. Such characterizations are crucial to understand arbitrage opportunities in models with fractional price processes (c.f.\@ the stickiness property in \cite{guasoni2006no,czichowsky2017portfolio} and the notion of arbitrage times in \cite{peyre2015no}). Moreover, the Markov property brings about a well-defined notion of the states of the model, which makes it possible to talk about calibration in a meaningful way.  

Presumably some of our results can be generalized to a non-Gaussian setting by replacing the Brownian noise by a L\'evy process. This would lead to a representation of fractional L\'evy processes \cite{marquardt2006fractional, engelke2013unifying, kluppelberg2015generalized} as superpositions of L\'evy-driven Ornstein--Uhlenbeck processes. It would also be interesting to derive similar representations for affine Volterra processes as in \cite{jaber2017affine}, where the affine structure is more important than in our Gaussian setting. 

The paper is structured as follows. In \autoref{sec:ou} we prove that the collection of OU processes is indeed a Banach-space valued affine process. In \autoref{sec:fbm} we deduce the affine representation of fractional Brownian motion. \autoref{sec:applications} is dedicated to applications in interest rate modeling and \autoref{sec:stein} to a fractional version of the stochastic volatility model of Stein and Stein \cite{stein1991}. Some auxiliary results and proofs are collected in \autoref{sec:auxiliary}.

\section{Infinite-dimensional Ornstein--Uhlenbeck process}\label{sec:ou}

\subsection{Setup and notation}\label{sec:setup}

Let $(\Omega,\mathcal F,(\mathcal F_t)_{t \in \mathbb R},\mathbb Q)$ be a filtered probability space satisfying the usual conditions, let $W=(W_t)_{t \in \mathbb R}$ be a two-sided $(\mathcal F_t)$-Brownian motion on $\Omega$, and let $\mathcal P$ denote the predictable sigma-algebra on $\Omega \times \mathbb R_+$.

\begin{definition}[OU processes]\label{def:ou}
Given a collection of $\mathcal F_0$-measurable $\mathbb R$-valued random variables $Y_0^x,Z_0^x$ indexed by $x \in (0,\infty)$, let for each $t\geq0$
\begin{align}\label{equ:ou}
Y_t^x&=Y_0^xe^{-tx} + \int_0^t e^{-(t-s)x}dW_s,
&
Z_t^x&=Z_0^xe^{-tx} + \int_0^t e^{-(t-s)x}Y_s^xds,
\end{align}
and let $Y_t=(Y^x_t)_{x>0}$ and $Z_t=(Z^x_t)_{x>0}$ denote the collection of OU processes indexed by the speed of mean reversion $x$. 
\end{definition}

\begin{remark}\label{rem:ou_sde}
For each $x \in (0,\infty)$, the process $(Y_t^x,Z_t^x)_{t\geq 0}$ solves the SDE
\begin{align}\label{equ:ou_sde}
dY_t^x &= -xY_t^xdt+dW_t,
& 
dZ_t^x&=(-xZ_t^x+Y_t^x)dt.
\end{align}
Therefore, it is a bi-variate OU process, and the variable $x$ is related to the speed of mean reversion of the process (see \autoref{lem:representation_ou} for details).
\end{remark}

\subsection{\texorpdfstring{Ornstein--Uhlenbeck process in $L^1$}{Ornstein--Uhlenbeck process in L1}}\label{sec:banach}

We show in this section that the process $(Y_t,Z_t)_{t\geq 0}$ takes values in  $L^1(\mu)\times L^1(\nu)$, where the measures $\mu$ and $\nu$ are subject to the following conditions.

\begin{assumption}[Integrability condition]
\label{ass:integrability1}\label{ass:density1}
$\mu$ and $\nu$ are sigma-finite measures on $(0,\infty)$ such that $\nu$ has a density $p$ with respect to $\mu$ and for each $t>0$,
\begin{align*}
&\int_0^\infty (1\wedge x^{-\frac12})\mu(dx)<\infty,
&
&\int_0^\infty (1\wedge x^{-\frac32}) \nu(dx)<\infty,
&
&\sup_{x\in(0,\infty)}p(x)e^{-tx}<\infty.	
\end{align*}
\end{assumption}

We endow the spaces $L^1(\mu)$, $L^1(\nu)$, and $L^1(\mu)\times L^1(\nu)$ with the norm topology and denote the corresponding Borel sigma-algebras by $\mathcal B$. Thus, a process $(X,Y)\colon\Omega\times\mathbb R_+\to L^1(\mu)\times L^1(\nu)$ is predictable if and only if it is $\mathcal P/\mathcal B(L^1(\mu)\times L^1(\nu))$-measurable, which is equivalent to $X$ being $\mathcal P/\mathcal B(L^1(\mu))$-measurable and $Y$ being $\mathcal P/\mathcal B(L^1(\nu))$-measurable because the norm topology of $L^1$ spaces has a countable basis \cite[Theorem~III.5.10]{elstrodt2011mass}. The pairing between the spaces $L^1(\mu)$ and $L^\infty(\mu)$ is denoted by $\langle \cdot,\cdot\rangle_\mu$, and similarly for $L^1(\nu)$ and $L^\infty(\nu)$. The complexification of these spaces is denoted by $L^1(\mu;\mathbb C)$, etc.

\begin{theorem}[OU process in $L^1$]\label{thm:banach}
Let $\mu,\nu$ satisfy \autoref{ass:integrability1} and let $(Y_0,Z_0)\in L^1(\mu)\times L^1(\nu)$. Then the process $(Y_t,Z_t)_{t\geq 0}$ has a predictable $L^1(\mu)\times L^1(\nu)$-valued version and is Gaussian.
\end{theorem}

\begin{remark}
Carmona and Coutin \cite{carmona1998fractional} show the weaker statement that for each fixed $t\geq 0$, the random variable $Y_t$ lies a.s.\ in $L^1(\mu)$. 
\end{remark}

\begin{proof}
It is shown in \autoref{lem:representation_ou} that for each $x\in(0,\infty)$ the process $(Y^x_t,Z^x_t)_{t\geq 0}$ can be represented as
\begin{align}\label{equ:representation_ou}
Y^x_t&=Y^x_0e^{-tx}+\int_0^te^{-(t-s)x}dW_s,
&
Z^x_t&=Z^x_0e^{-tx}+Y^x_0te^{-tx}+\int_0^t(t-s)e^{-(t-s)x}dW_s.
\end{align}
By \autoref{ass:density1}, the deterministic parts in the above representation are $L^1(\mu)$- and $L^1(\nu)$-valued continuous functions, respectively. Thus, we can assume without loss of generality that $Y_0$ and $Z_0$ are zero. 

In \autoref{lem:integrability_ou} it is shown that for each fixed $t\geq 0$, $(Y_t,Z_t)\in L^1(\mu)\times L^1(\nu)$ holds almost surely. Moreover, for any $(u,v) \in L^\infty(\mu)\times L^\infty(\nu)$, the random variables $\langle Y_t,u\rangle_\mu$ and $\langle Z_t,v\rangle_\nu$ are centered Gaussian, as shown in \autoref{lem:representation_inf_dim_ou}. Let $P_t\colon L^\infty(\mu)\to L^1(\mu)$ and $Q_t\colon L^\infty(\mu)\to L^1(\nu)$ be the associated covariance operators, which are calculated explicitly in \autoref{lem:cov}. 

To show that $Y$ has a predictable $L^1(\mu)$-valued version, let $T \in (0,\infty)$, and let $\mathsf H_T \subseteq L^1(\mu)$ be the reproducing kernel Hilbert space of $P_T$ (see \autoref{sec:reproducing}). The inclusion of $\mathsf H_T$ in $L^1(\mu)$ is $\gamma$-radonifying (see \autoref{def:gamma:gamma}) because $Y_T$ provides an instance of a Gaussian random variable with covariance operator $P_T$ (see  \autoref{thm:gamma:gamma}). For each $s\in (0,T]$ define $\Theta_1(s)\in L^1(\mu)$ and $\Theta_1^\ast(s)\colon L^\infty(\mu)\to\mathbb R$ by
\begin{align}\label{equ:Theta1}
\Theta_1(s)(x)&=e^{-sx}, &
\Theta_1^\ast(s)(u)&=\left\langle\Theta_1(s),u\right\rangle_\mu.
\end{align}
Then $\Theta_1^\ast$ satisfies for each $t\in [0,T]$ and $u\in L^\infty(\mu)$
\[
\int_0^t \big(\Theta_1^\ast(t-s)(u)\big)^2 ds=\int_0^t\left(\int_0^\infty e^{-x(t-s)}u(x)\mu(dx)\right)^2ds=\langle P_t u,u\rangle_\mu<\infty,
\]
where the order of integration can be exchanged because condition \eqref{equ:fubini1} is satisfied by \autoref{equ:elem_int15}. By \autoref{thm:gamma:ou} there exists a predictable process $\widetilde Y\colon \Omega\times[0,T] \to L^1(\mu)$ which satisfies for all $s,t \in [0,T]$ and all $u,v \in L^\infty(\mu)$ that 
\begin{equation*}
\mathbb E\left[\langle \widetilde Y_t,u\rangle_\mu \langle \widetilde Y_s,v\rangle_\mu\right]
=
\int_0^{t\wedge s} \langle \Theta_1(t-r)^*u,\Theta_1(s-r)^*v\rangle_{\mathsf H_T} dr.
\end{equation*}
By \autoref{thm:gamma:ou} this equation determines $\widetilde Y$ uniquely up to modifications. As this equation is satisfied by $Y$, and $T \in (0,\infty)$ is arbitrary, we have shown that $Y$ has a predictable $L^1(\mu)$-valued version. 

We use the same argument to show that $Z$ has a predictable, $L^1(\nu)$-valued version. This time, for each $T \in (0,\infty)$, let $\mathsf H_T$ be the reproducing kernel Hilbert space of $Q_T$ (see \autoref{sec:reproducing}). Then the embedding of $\mathsf Q_T$ into $L^1(\nu)$ is $\gamma$-radonifying (see \autoref{def:gamma:gamma}) because $Z_T$ provides an instance of a Gaussian random variable with covariance operator $Q_T$ (see \autoref{thm:gamma:gamma}). For each $s \in (0,T]$ define $\Theta_2(s)\in L^1(\nu)$ and $\Theta_2^\ast(s)\colon L^\infty(\nu)\to\mathbb R$ by
\begin{align}\label{equ:Theta2}
\Theta_2(s)(x)&=se^{-sx}, &
\Theta_2^\ast(s)(v)&=\langle\Theta_2(s),v\rangle_\nu.
\end{align}
Then $\Theta_2^\ast$ satisfies for each $t\in [0,T]$ and $v\in L^\infty(\nu)$
\[
\int_0^t \big(\Theta_2^\ast(t-s)(v)\big)^2 ds
=\int_0^t\left(\int_0^\infty (t-s)e^{-x(t-s)}v(x)\nu(dx)\right)^2ds
=\langle Q_tv,v\rangle_\nu<\infty,
\]
where the order of integration can be exchanged because condition \eqref{equ:fubini1} is satisfied by \autoref{equ:elem_int16}. By the same argument as above $Z$ has a predictable version $\widetilde Z\colon \Omega\times\mathbb R_+\to L^1(\nu)$.
\end{proof}

\subsection{Affine structure}\label{sec:affine}

We derive an infinite-dimensional affine transformation formula for the conditional exponential moments of $\left\langle Y, u\right\rangle_{\mu}$ and $\left\langle Z, v\right\rangle_{\nu}$ for test functions $u\in L^\infty(\mu;\mathbb C)$ and $v\in L^\infty(\nu;\mathbb C)$.

\begin{theorem}[Affine structure]\label{thm:affine}
Let $\mu,\nu$ satisfy \autoref{ass:integrability1} and let $(Y_0,Z_0)\in L^1(\mu)\times L^1(\nu)$. Then the process $(Y,Z)$ is affine in the sense that for each $0\leq t\leq T$ and $(u,v) \in L^\infty(\mu;\mathbb C)\times L^\infty(\nu;\mathbb C)$, the relation
\[
\mathbb E\left[e^{\left\langle Y_T,u\right\rangle_\mu+\left\langle Z_T,v\right\rangle_\nu}\middle|\mathcal F_t\right] =e^{\phi_0(T-t,u,v)+\left\langle Y_t,\phi_1(T-t,u,v)\right\rangle_\mu+\left\langle Z_t,\phi_2(T-t,u,v)\right\rangle_\nu}
\]
holds with probability one, where the functions
\begin{equation}\label{equ:def_phipsi}
(\phi_0,\phi_1,\phi_2)\colon[0,\infty)\times L^\infty(\mu;\mathbb C)\times L^\infty(\nu;\mathbb C) \to \mathbb C \times L^\infty(\mu;\mathbb C) \times L^\infty(\nu;\mathbb C)
\end{equation}
are given by
\begin{equation}\label{equ:phipsi}
\begin{aligned}
\phi_0(\tau,u,v)&=\frac{1}{2}\int_0^\tau \left(\int_0^\infty \phi_1(s,u,v)(x)\mu(dx)\right)^2 ds,
\\
\phi_1(\tau,u,v)(x)&= e^{-\tau x} \big(  u(x)+ \tau v(x)p(x)\big),
\\
\phi_2(\tau,u,v)(x)&= e^{-\tau x}v(x).
\end{aligned}
\end{equation}
\end{theorem}

\begin{proof}
\autoref{lem:representation_inf_dim_ou} states that for each $0\leq t\leq T$, the random variable $\left\langle Y_T,u\right\rangle_\mu+\left\langle Z_T,v\right\rangle_\nu$ is Gaussian, given $\mathcal F_t$, with mean 
\[
\begin{aligned}
&\int_0^\infty Y^x_t e^{-(T-t)x} u(x)\mu(dx)+\int_0^\infty \left(Z^x_t e^{-x(T-t)} + Y^x_t (T-t) e^{-x(T-t)}\right)v(x)\nu(dx)\\
&\qquad=
\left\langle Y_t,\phi_1(T-t,u,v)\right\rangle_\mu+\left\langle Z_t,\phi_2(T-t,u,v)\right\rangle_\nu.
\end{aligned}
\]
By It\=o's isometry, the conditional variance of $\left\langle Y_T,u\right\rangle_\mu+\left\langle Z_T,v\right\rangle_\nu$ given $\mathcal F_t$ is
\[
\int_t^T \left(\int_0^\infty e^{-(T-s)x}u(x)\mu(dx)
+\int_0^\infty (T-s) e^{-x(T-s)}v(x)\nu(dx)\right)^2 ds, 
\]
which equals $2 \phi_0(T-t,u,v)$. Thus, 
\begin{equation*}\begin{aligned}
\mathbb E\left[e^{\left\langle Y_T,u\right\rangle_\mu+\left\langle Z_T,v\right\rangle_\nu}\middle|\mathcal F_t\right]&=e^{\frac{1}{2}\operatorname{Var}\left(\left\langle Y_T,u\right\rangle_\mu+\left\langle Z_T,v\right\rangle_\nu|\mathcal F_t\right)+\mathbb E\left[\left\langle Y_T,u\right\rangle_\mu+\left\langle Z_T,v\right\rangle_\nu|\mathcal F_t\right]}\\&=e^{\phi_0(T-t,u,v)+\left\langle Y_t,\phi_1(T-t,u,v)\right\rangle_\mu+\left\langle Z_t,\phi_2(T-t,u,v)\right\rangle_\nu}.
\end{aligned}
\end{equation*}
This concludes the proof.
\end{proof}

The coefficient functions $(\phi_0,\phi_1,\phi_2)$ are solutions of an infinite-dimensional system of Riccati equations. To formulate the equations, we need to introduce some topology. We endow the spaces $L^\infty(\mu;\mathbb C)$ and $L^\infty(\nu;\mathbb C)$ with the weak-star topology. Then they are locally convex separable Hausdorff vector spaces. In particular, differentiability of curves with values in these spaces is well-defined. 

\begin{definition}[Riccati equations]\label{def:riccati}
Mappings $\phi_0,\phi_1,\phi_2$ as in \eqref{equ:def_phipsi} are called solutions of the Riccati equations if they are continuous in $t$ on the interval $[0,\infty)$, differentiable in $t$ on the interval $(0,\infty)$, and satisfy
\begin{equation}\label{equ:riccati}
\begin{aligned}
\partial_\tau (\phi_0,\phi_1,\phi_2)(\tau,u,v)&= (R_0,R_1,R_2)\big(\phi_1(\tau,u,v),\phi_2(\tau,u,v)\big),
\\
(\phi_0,\phi_1,\phi_2)(0,u,v)&=(0,u,v),
\end{aligned}
\end{equation}
where the mappings
\begin{equation*}
(R_0,R_1,R_2)\colon L^\infty(\mu;\mathbb C) \times L^\infty(\nu;\mathbb C) \to \mathbb C \times L^0(\mu;\mathbb C) \times L^0(\nu;\mathbb C)
\end{equation*} 
are given by
\[
\begin{aligned}
R_0(u,v)&= \frac12 \left(\int_0^\infty u(x)\mu(dx)\right)^2,\\
R_1(u,v)(x)&= -x u(x) + p(x) v(x),\\
R_2(u,v)(x)&= -x v(x).
\end{aligned}
\]
\end{definition}

\begin{lemma}[Riccati equations]\label{lem:riccati}
The functions $(\phi_0,\phi_1,\phi_2)$ defined in \autoref{equ:phipsi} are the unique solution of the Riccati equations \eqref{equ:riccati}.
\end{lemma}

\begin{proof}
It is straightforward to verify that the functions $(\phi_0,\phi_1,\phi_2)$ given by \autoref{equ:phipsi} solve the Riccati equations in the sense of \autoref{def:riccati}. Let $(\overline\phi_0,\overline\phi_1,\overline\phi_2)$ be any other solution. Then $e^{xt}(\phi_2-\overline\phi_2)$ has vanishing  derivative and initial condition, implying that it is constant and $\phi_2=\overline\phi_2$. The same applies to $e^{xt}(\phi_1-\overline\phi_1)$, showing that $\phi_1=\overline\phi_1$, and to $\phi_0-\overline\phi_0$, showing that $\phi_0=\overline\phi_0$. 
\end{proof}

\subsection{Continuity of sample paths}\label{sec:paths}

Under the following conditions on the measures $\mu$ and $\nu$, the process $(Y,Z)$ has continuous sample paths in $L^1(\mu)\times L^1(\nu)$ with respect to the norm topology.

\begin{assumption}[Integrability condition]\label{ass:integrability2}
$\mu$ and $\nu$ are sigma-finite measures on $(0,\infty)$ such that $\nu$ has a density $p$ with respect to $\mu$ and for each $t>0$
\begin{align*}
&\int_0^\infty\log(1+tx)^{1/2}x^{-\frac{1}{2}}\mu(dx)<\infty,
&
&\int_0^\infty\log(1+tx)^{1/2}x^{-\frac{3}{2}}\nu(dx)<\infty,
&
\sup_{x\in(0,\infty)}p(x)e^{-tx}<\infty.
\end{align*}
\end{assumption}

\begin{remark}
Compared to \autoref{ass:integrability1}, \autoref{ass:integrability2} is equivalent near zero and stronger near infinity, as can be seen from the limits
\[
\forall t>0\colon\quad\lim_{x\to 0^+}\frac{\log(1+tx)^{1/2}x^{-\frac{1}{2}}}{1\wedge x^{-\frac{1}{2}}}=\sqrt{t},\quad \lim_{x\to\infty}\frac{\log(1+tx)^{1/2}x^{-\frac{3}{2}}}{1\wedge x^{-\frac{3}{2}}}=\infty.
\]
\end{remark}

\begin{theorem}[Continuity of sample paths]\label{thm:paths}
Under \autoref{ass:integrability2}, the process $(Y,Z)$ has continuous sample paths in $L^1(\mu)\times L^1(\nu)$ if the initial condition $(Y_0,Z_0)$ lies in this space.
\end{theorem}

\begin{remark}
Note that \autoref{thm:paths} does not guarantee that $(Y,Z)$ is a Gaussian process in $L^1(\mu)\times L^1(\nu)$; this follows from \autoref{thm:banach} under \autoref{ass:integrability1}.
\end{remark}

\begin{proof}
The expressions $Y^x_0 e^{-tx}$ and $Z^x_0 e^{-tx}+Y_0^xte^{-tx}$ define continuous $L^1(\mu)$- and $L^1(\nu)$-valued functions, respectively. Thus, it follows from the representation of $(Y,Z)$ in \autoref{equ:representation_ou} that we may assume $(Y_0,Z_0)=0$ without loss of generality. 

By \autoref{lem:ou_max_ineq}, and \autoref{ass:integrability2} on $\mu$, integration with respect to $\mu$ yields
\[
\mathbb E\left[\int_0^\infty\sup_{s\in[0,t]} |Y^x_s| \mu(dx)\right]\leq C\int_0^\infty \log(1+tx)^{1/2} x^{-\frac{1}{2}}\mu(dx)<\infty,
\]
where we are allowed to exchange the order of integration since the integrand is positive. This implies that $\mathbb Q[\forall t\colon Y_t \in L^1(\mu)]=1$. Moreover, by the dominated convergence theorem with the $\sup$ process of $Y$ as majorant, $\mathbb Q[Y\in C([0,\infty);L^1(\mu))]=1$. 

For the process $Z$, the estimate of \autoref{lem:ou_max_ineq} and \autoref{ass:integrability2} on $\nu$ show that $\mathbb Q[\forall t\colon Z^x_t \in L^1(\nu)]=1$. As before, the dominated convergence theorem with the $\sup$ process of $Z$ as majorant implies $\mathbb Q[Z\in C([0,\infty);L^1(\nu))]=1$.
\end{proof}

\subsection{Semimartingale property}\label{sec:semimartingale}

In this section we investigate under which conditions linear functionals of the process $(Y,Z)$ are semimartingales. We consider time-dependent linear functionals as this will be needed later in applications.
  
\begin{theorem}[Semimartingale property]\label{thm:semimartingale}
Let \autoref{ass:integrability1} be in place. Let $f^x_t$ and $g^x_t$ be real-valued, deterministic, jointly measurable in $(x,t)\in(0,\infty)\times[0,\infty)$, differentiable in $t$ and satisfy
\[
\forall t\geq0\colon\|f_t\|_{L^\infty(\mu)}<\infty\:\:\text{and}\:\:\|g_t\|_{L^\infty(\nu)}<\infty.
\]
Assume $(Y_0,Z_0)\in L^1(\mu)\times L^1(\nu)$, a.s., and for each $t\geq0$
\begingroup\allowdisplaybreaks\begin{align}
\label{equ:semimartingale1}
&\int_0^\infty \int_0^t|\partial_s f^x_s-xf^x_s|(1\wedge x^{-\frac12})ds\mu(dx)<\infty,\\
\label{equ:semimartingale2}
&\int_0^\infty \sqrt{\int_0^t (f^x_s)^2 ds}\mu(dx)<\infty,\\
\label{equ:semimartingale3}
&\int_0^\infty \int_0^t|\partial_s g^x_s-xg^x_s|(1\wedge x^{-\frac32})ds\nu(dx)<\infty,\\ 
\label{equ:semimartingale4}
&\int_0^\infty \int_0^t|g^x_s|(1\wedge x^{-\frac12})ds\nu(dx)<\infty.
\end{align}\endgroup
Then $(\langle Y_t,f_t\rangle_\mu)_{t\geq0}$ and $(\langle Z_t,g_t\rangle_\nu)_{t\geq0}$ are semimartingales with decompositions
\begingroup\allowdisplaybreaks\begin{equation}\label{equ:semimartingale_decompositions}
\begin{aligned}
\left\langle Y_t,f_t\right\rangle_\mu&=\left\langle Y_0,f_0\right\rangle_\mu
+\int_0^t\int_0^\infty \left(\partial_sf^x_s-xf^x_s\right)Y^x_s \mu(dx)ds+\int_0^t \int_0^\infty f^x_s \mu(dx)dW_s,\\
\left\langle Z_t,g_t\right\rangle_\nu&= \left\langle Z_0,g_0\right\rangle_\nu+\int_0^t\int_0^\infty\left(\partial_s g^x_s-xg^x_s\right)Z^x_s\nu(dx)ds+\int_0^t\int_0^\infty g^x_sY^x_s\nu(dx)ds.
\end{aligned}
\end{equation}\endgroup
\end{theorem}

\begin{proof}
First observe that 
\[
\begin{aligned}
\left\langle Y_t,f_t\right\rangle_\mu&=\left\langle Y_t-Y^x_0e^{-xt},f_t\right\rangle_\mu+\left\langle Y^x_0e^{-xt},f_t\right\rangle_\mu,\\
\left\langle Z_t,g_t\right\rangle_\nu&=\left\langle Z_t-Z^x_0e^{-xt}-Y^x_0te^{-xt},g_t\right\rangle_\nu+\left\langle Z^x_0e^{-xt},g_t\right\rangle_\nu+\left\langle Y^x_0te^{-xt},g_t\right\rangle_\nu.
\end{aligned}
\]
Since $\left\langle Y^x_0e^{-xt},f_t\right\rangle_\mu$, $\left\langle Z^x_0e^{-xt},g_t\right\rangle_\nu$ and $\left\langle Y^x_0te^{-xt},g_t\right\rangle_\nu$ are finite variation processes we assume without loss of generality that $Y_0=Z_0=0$.
By SDE \eqref{equ:ou_sde} for $(Y,Z)$ and It\=o's formula, the semimartingale decomposition of the process $(f^x_tY^x_t,g^x_tZ^x_t)$ is given by
\[
\begin{aligned}
f^x_tY^x_t &=\int_0^t \left(\partial_s f^x_s-xf^x_s\right)Y^x_s ds 
+ \int_0^t f^x_s dW_s,\\
g^x_tZ^x_t&=\int_0^{t}\left(\partial_sg^x_s-xg^x_s\right)Z^x_sds+\int_0^tg^x_sY^x_sds.
\end{aligned}
\]
Therefore, 
\[
\begin{aligned}
\left\langle Y_t,f_t\right\rangle_\mu&=\int_0^\infty\int_0^t \left(\partial_s f^x_s-xf^x_s\right)Y^x_s ds \mu(dx)+\int_0^\infty \int_0^t f^x_sdW_s\mu(dx),\\ \left\langle Z_t,g_t\right\rangle_\nu&=\int_0^\infty\int_0^t \left(\partial_s g^x_s-xg^x_s\right)Z^x_s ds \nu(dx)+\int_0^\infty\int_0^tg^x_sY^x_s ds \nu(dx).
\end{aligned}
\]
By \autoref{thm:fubini} one obtains the semimartingale decompositions of $\left\langle Y_t,f_t\right\rangle_\mu$ and $\left\langle Z_t,g_t\right\rangle_\nu$. By \autoref{lem:semimartingale_fubini} and Equations \eqref{equ:semimartingale1}-\eqref{equ:semimartingale4} conditions \eqref{equ:fubini1} and \eqref{equ:fubini2} are satisfied.
\end{proof}

\subsection{Stationary distribution}\label{sec:sta}

We show that the stationary distribution of $(Y,Z)$ is in general not a Gaussian distribution on $L^1(\mu)\times L^1(\nu)$, but only on a larger space $L^1(\mu_\infty)\times L^1(\nu_\infty)$ corresponding to stronger integrability conditions on the measures $\mu_\infty$ and $\nu_\infty$.

\begin{assumption}[Integrability condition]\label{ass:integrability_inf}
$\mu_\infty,\nu_\infty$ are sigma-finite measures on $(0,\infty)$ such that $\nu_\infty$ has a density $p_\infty$ with respect to $\mu_\infty$ and
\begin{align*}
&\int_0^\infty x^{-1/2} \mu_\infty(dx)<\infty,
&
&\int_0^\infty x^{-3/2} \nu_\infty(dx)<\infty,
&
&\sup_{x\in(0,\infty)}p_\infty(x)e^{-tx}<\infty.
\end{align*}
\end{assumption}

\begin{remark}
\autoref{ass:integrability_inf} is more stringent than \autoref{ass:integrability1}. The difference is the decay of the measures near zero: $\mu,\nu$ satisfy \autoref{ass:integrability1} if and only if the measures
\begin{align}\label{equ:mu_nu_inf}
\mu_\infty(dx)&=(1\wedge x^{1/2})\mu(dx), 
&
\nu_\infty(dx)&=(1\wedge x^{1/2})\nu(dx)
\end{align} satisfy \autoref{ass:integrability_inf}. In this case, $L^1(\mu)\times L^1(\nu) \subset L^1(\mu_\infty)\times L^1(\nu_\infty)$.
\end{remark}

\begin{theorem}[Stationary distribution]\label{thm:sta}
The random variables $Y_\infty=(Y^x_\infty)_{x>0}$ and $Z_\infty=(Z_\infty^x)_{x>0}$ defined by
\begin{align}\label{equ:yz_inf}
Y^x_\infty&=\int_{-\infty}^0 e^{sx}dW_s,
&
Z^x_\infty=-\int_{-\infty}^0 se^{xs}dW_s
\end{align}
are normally distributed on $L^1(\mu_\infty)\times L^1(\nu_\infty)$. Their distribution is stationary in the sense that $(Y_t,Z_t)$ is equal in distribution to $(Y_\infty,Z_\infty)$ if $(Y_0,Z_0)$ is equal in distribution to $(Y_\infty,Z_\infty)$. 
\end{theorem}

\begin{proof}
$(Y_\infty,Z_\infty) \in L^1(\mu_\infty)\times L^1(\nu_\infty)$ holds almost surely because
\begin{align*}
\mathbb E\left[\|Y_\infty\|_{L^1(\mu_\infty)}\right]
&=
\int_0^\infty \mathbb E\left[\left|\int_{-\infty}^0 e^{sx}dW_s\right|\right] \mu_\infty(dx)
= \int_0^\infty \sqrt{\frac{1}{\pi x}} \mu_\infty(dx) <\infty,
\\
\mathbb E\left[\|Z_\infty\|_{L^1(\nu_\infty)}\right]
&=
\int_0^\infty \mathbb E\left[\left|\int_{-\infty}^0 se^{sx}dW_s\right|\right] \nu_\infty(dx)
= \int_0^\infty \sqrt{\frac{1}{2\pi x^3}} \nu_\infty(dx) <\infty.
\end{align*}
For each $u,v \in L^\infty(\mu_\infty)\times L^\infty(\nu_\infty)$, the random variable $\left\langle Y_\infty,u\right\rangle_{\mu_\infty}+\left\langle Z_\infty,v\right\rangle_{\nu_\infty}$ can be expressed by Fubini (\autoref{thm:fubini}) as
\begin{equation}\label{equ:yz_inf_representation}
\begin{aligned}
\left\langle Y_\infty,u\right\rangle_{\mu_\infty}+\left\langle Z_\infty,v\right\rangle_{\nu_\infty}&=\int_{-\infty}^0 \int_0^\infty e^{sx}u(x)\mu_\infty(dx)dW_s+\int_{-\infty}^0 \int_0^\infty se^{sx}v(x)\nu_\infty(dx)dW_s.
\end{aligned}
\end{equation}
Condition \eqref{equ:fubini2} of Fubini's theorem is satisfied because 
\begin{align*}
&\int_0^\infty \sqrt{\int_{-\infty}^0 e^{2sx} u(x)^2ds}\mu_\infty(dx)
\leq
\|u\|_{L^\infty(\mu_\infty)}
\int_0^\infty \sqrt{\frac{1}{2x}}\mu_\infty(dx) < \infty,
\\
&\int_0^\infty \sqrt{\int_{-\infty}^0 s^2e^{2sx} v(x)^2ds}\nu_\infty(dx)
\leq
\|v\|_{L^\infty(\mu_\infty)}
\int_0^\infty \sqrt{\frac{1}{4x^3}}\nu_\infty(dx) < \infty.
\end{align*}
Therefore, $\left\langle Y_\infty,u\right\rangle_{\mu_\infty}+\left\langle Z_\infty,v\right\rangle_{\nu_\infty}$ is a centered Gaussian random variable on $L^1(\mu_\infty)\times L^1(\nu_\infty)$. To show that the distribution of $(Y_\infty,Z_\infty)$ is stationary, let us assume that $(Y_0,Z_0)=(Y_\infty,Z_\infty)$. Then \autoref{lem:representation_ou} implies
\begin{align*}
Y^x_t&=\int_{-\infty}^t e^{-(t-s)x}dW_s,
&
Z^x_t&=\int_{-\infty}^t (t-s)e^{-(t-s)x}dW_s,
\end{align*}
which is equal in distribution to $Y_\infty$ and $Z_\infty$, respectively. 
\end{proof}

\begin{theorem}[Convergence to the stationary distribution]
For any initial condition $(Y_0,Z_0) \in L^1(\mu_\infty)\times L^1(\nu_\infty)$ and any $t\geq 0$, we consider $(Y_t,Z_t)$ as a random variable with values in the space $L^1(\mu_\infty)\times L^1(\nu_\infty)$, which we endow with the weak topology. Then $(Y_t,Z_t)$ converges in distribution to $(Y_\infty,Z_\infty)$ as $t\to\infty$.
\end{theorem}

\begin{proof}
Let $(u,v)\in L^\infty(\mu_\infty)\times L^\infty(\nu_\infty)$. By \autoref{equ:yz_inf_representation} and It\=o's isometry the variance of the centered Gaussian random variable $\langle Y_\infty,u\rangle_{\mu_\infty} + \langle Z_\infty,v\rangle_{\nu_\infty}$ is 
\begin{align*}
\mathbb E\left[\left(\langle Y_\infty,u\rangle_{\mu_\infty}+\langle Z_\infty,v\rangle_{\nu_\infty}\right)^2\right] 
= \int_{-\infty}^0 \left(\int_0^\infty e^{sx}\big(u(x)+sv(x)p_\infty(x)\big)\mu_\infty(dx)\right)^2 ds.
\end{align*}
Assume for a moment that $(Y_0,Z_0)=0$. As the measures $\mu_\infty$ and $\nu_\infty$ satisfy the conditions of \autoref{thm:affine}, 
\begin{align*}
\lim_{t\to\infty}\mathbb E\left[e^{\left\langle Y_t,u\right\rangle_{\mu_\infty}+\left\langle Z_t,v\right\rangle_{\nu_\infty}}\right]
&=
\lim_{t\to\infty}e^{\phi_0(t,u,v)+\left\langle Y_0,\phi_1(t,u,v)\right\rangle_{\mu_\infty}+\left\langle Z_0,\phi_2(t,u,v)\right\rangle_{\nu_\infty}}
\\&=
e^{\frac12 \int_0^\infty \left(\int_0^\infty e^{-sx}(u(x)+sv(x)p_\infty(x))\mu_\infty(dx)\right)^2 ds} 
\\&= 
e^{\frac12 \operatorname{Var}\left(\langle Y_\infty,u\rangle_{\mu_\infty}+\langle Z_\infty,v\rangle_{\nu_\infty}\right)}=\mathbb E\left[e^{\langle Y_\infty,u\rangle_{\mu_\infty}+\langle Z_\infty,v\rangle_{\nu_\infty}}\right].
\end{align*}
This shows point-wise convergence of the characteristic functions of $(Y_t,Z_t)$ to the characteristic functions of $(Y_\infty,Z_\infty)$. By \autoref{lem:tightness} the laws of the random variables $(Y_t,Z_t)_{t\geq 0}$ are tight on the space $L^1(\mu_\infty)\times L^1(\nu_\infty)$ with the weak topology. It follows that $(Y_t,Z_t)$ converges in distribution on $L^1(\mu_\infty)\times L^1(\nu_\infty)$ to $(Y_\infty,Z_\infty)$ (see e.g.\ \cite[Theorem~9]{dunford1958linear}). 

To account for arbitrary initial conditions $(Y_0,Z_0) \in L^1(\mu_\infty)\times L^1(\nu_\infty)$, we need to add the deterministic functions $Y_0^x e^{-tx}$ and $Z_0^x e^{-tx} + Y_0^x t e^{-tx}$ to the processes $Y^x_t$ and $Z^x_t$ considered above (see \autoref{lem:representation_ou}). For $t\to\infty$, these functions converge to zero in the corresponding $L^1$ spaces. It follows that convergence in distribution to $(Y_\infty,Z_\infty)$ holds regardless of the initial condition. 
\end{proof}

\subsection{\texorpdfstring{Ornstein--Uhlenbeck process with values in $L^2$}{Ornstein--Uhlenbeck process with values in L2}}

We defined $(Y,Z)$ as an $L^1$-valued process because the construction of fractional Brownian motion in \autoref{sec:fbm} involves a pairing of $(Y,Z)$ with the constant function $1$. Nevertheless, it is good to know that $(Y,Z)$ can also be understood as an $L^2$-valued process.

\begin{assumption}[Integrability condition]\label{ass:l2}
$\mu$ and $\nu$ are sigma-finite measures on $(0,\infty)$ such that $\nu$ has a density $p$ with respect to $\mu$ and for each $t>0$, 
\begin{align*}
&\int_0^\infty (1\wedge x^{-1})\mu(dx)<\infty,
&\int_0^\infty (1\wedge x^{-3})\nu(dx)<\infty,
&\sup_{x\in(0,\infty)} e^{-tx}p(x)<\infty.
\end{align*}
\end{assumption}

\begin{theorem}[OU process in $L^2$]\label{thm:hilbert}
Let $\mu,\nu$ satisfy \autoref{ass:l2} and let $(Y_0,Z_0)\in L^2(\mu)\times L^2(\nu)$. Then the process $(Y_t,Z_t)_{t\geq 0}$ has a predictable $L^2(\mu)\times L^2(\nu)$-valued version and is a Gaussian affine process on $L^2(\mu)\times L^2(\nu)$.
\end{theorem}

The theorem can be proven along the lines of \autoref{thm:banach,thm:affine}. Here we present an alternative proof, which uses the theory of Hilbert-space valued stochastic convolutions.

\begin{proof}
We want to construct the stochastic convolutions in \autoref{equ:representation_ou} as $L^2(\mu)$- and $L^2(\nu)$-valued processes, respectively. The setting of \cite[Sections~5.1.1--5.1.2]{daPrato2014se} does not apply directly because the volatility, which is the constant function $1$, does not belong to $L^2(\mu)$. Nevertheless, we can adapt the arguments of \cite[Theorem~5.2 and Proposition~3.6]{daPrato2014se} to our setting. For each $t \in [0,\infty)$ let 
\begin{equation*}
B_t\colon \mathbb R \to L^2(\mu), \qquad 
u \mapsto  (x\mapsto e^{-tx}u).
\end{equation*}
Then the $L^2(\mu)$-valued convolution $\int_0^t B_{t-s}dW_s$ exists by \cite[Theorem~5.2]{daPrato2014se} because
\begin{align*}
\int_0^t \|B_s\|_{HS(\mathbb R,L^2(\mu))}^2 ds
&=
\int_0^t \|B_s(1)\|_{L^2(\mu)}^2 ds
=
\int_0^t \int_0^\infty e^{-2sx} \mu(dx) ds
=
\int_0^\infty \frac{1-e^{-2tx}}{2x} \mu(dx)<\infty
\end{align*}
by \autoref{equ:elem_ineq4} and \autoref{ass:l2}, where $HS$ denotes the Hilbert--Schmidt operators. The convolution is mean-square continuous by the same arguments as in the proof of \cite[Theorem~5.2]{daPrato2014se}. Therefore, it is predictable \cite[Proposition~3.6]{daPrato2014se}. Similarly, it can be shown that $Z$ has a predictable, $L^2(\nu)$-valued version. The affine structure can be derived as in \autoref{sec:affine}.
\end{proof}

\begin{assumption}[Integrability condition]\label{ass:integrability_l2_stronger}
$\mu$ and $\nu$ are sigma-finite measures on $(0,\infty)$ such that $\nu$ has a density $p$ with respect to $\mu$. There is $\epsilon \in (0,1)$ such that for each $t>0$, 
\begin{align*}
&\int_0^\infty (1\wedge x^{-1+\epsilon})\mu(dx)<\infty,
&\int_0^\infty (1\wedge x^{-3+\epsilon})\nu(dx)<\infty,
&\sup_{x\in(0,\infty)} e^{-tx}p(x)<\infty.
\end{align*}
\end{assumption}

\begin{theorem}[Continuity of sample paths]
Under \autoref{ass:integrability_l2_stronger}, the process $(Y,Z)$ has continuous sample paths in $L^2(\mu)\times L^2(\nu)$ if the initial condition $(Y_0,Z_0)$ lies in this space.
\end{theorem}

\begin{proof}
Let $B$ be as in the proof of \autoref{thm:hilbert}. Then the estimate 
\begin{align*}
\int_0^t s^{-\epsilon}\|B_s\|_{HS(\mathbb R,L^2(\mu))}^2 ds
&=
\int_0^t \int_0^\infty s^{-\epsilon}e^{-2sx}\mu(dx) ds
\\&\leq
\int_0^\infty \left(2^{\epsilon-1}\Gamma(1-\epsilon)\vee \frac{t^{1-\epsilon}}{1-\epsilon}\right)\left(1\wedge x^{\epsilon-1}\right)\mu(dx)<\infty
\end{align*}
holds by \autoref{equ:elem_ineq7} for $\epsilon\in(0,1)$ as in \autoref{ass:integrability_l2_stronger}. Therefore, \cite[Theorem 5.11]{daPrato2014se} may be applied, showing that $Y$ has continuous sample paths in $L^2(\mu)$. (While the stochastic convolution $Y$ is not covered by the setting of \cite[Section~5.1.1--5.1.2]{daPrato2014se}, the same arguments as in the proof of \autoref{thm:hilbert} show that \cite[Theorem 5.11]{daPrato2014se} holds.) Similarly, it may be shown that the process $Z$ given by \autoref{equ:representation_ou} has continuous sample paths in $L^2(\nu)$.
\end{proof}

\subsection{Smoothness in the spatial dimension}\label{sec:smoothness}

We show in the following theorem that $(Y^x_t,Z^x_t)$ varies smoothly in $x$. To this aim, we extend \autoref{def:ou} of $(Y^x_t,Z^x_t)$ to  $x\leq0$ in the obvious way. The space $C^k(\mathbb R)$, $k \in \mathbb N \cup \{\infty\}$, is the Fr\'echet space with the topology of uniform convergence of derivatives up to order $k$ on compact sets. 

\begin{theorem}[Smoothness in the spacial dimension]\label{thm:smoothness}
For each $k \in \mathbb N \cup \{\infty\}$ and initial value $(Y_0,Z_0)\in C^k(\mathbb R)\times C^k(\mathbb R)$, the process $(Y,Z)$ is a Gaussian process on $C^k(\mathbb R)^2$ with continuous sample paths.
\end{theorem}

\begin{proof}
The deterministic parts in \autoref{equ:representation_ou} are smooth in $t$ and $x$. We set them to zero by assuming without loss of generality that $(Y_0,Z_0)=0$. By partial integration, the stochastic integrals in \autoref{equ:representation_ou} can be transformed into Lebesgue integrals:
\begin{align*}
Y^x_t &= W_t - \int_0^t W_s x e^{-(t-s)x}ds,
&
Z^x_t &= \int_0^t W_s\left(e^{-(t-s)x}-(t-s)xe^{-(t-s)x}\right)ds.
\end{align*}
The integrands, seen as functions of $(s,t)$, are continuous with values in $C^\infty(\mathbb R)$. This shows that $(Y_t,Z_t)_{t\geq 0}$ has continuous sample paths in $C^\infty(\mathbb R)^2$. The $k$-th spatial derivative, expressed as a stochastic integral, is given by
\begin{align*}
\partial_x^k Y^x_t &= \int_0^t (s-t)^k e^{-(t-s)x} dW_s, 
&
\partial_x^k Z^x_t &= -\int_0^t (s-t)^{k+1} e^{-(t-s)x} dW_s.
\end{align*}

To show that $(Y,Z)$ is a Gaussian process, it suffices to test with linear functionals on $C^k([-K,K])$ for $K \in \mathbb N$. By the Riesz representation theorem, the dual of $C^k([-K,K])$ is $\mathbb R^k\times \mathcal M([-K,K])$, where $\mathcal M$ stands for the space of signed regular Borel measures endowed with the total variation norm \cite[\mbox{}IV.13.36]{dunford1958linear}. The pairing of $Y_t \in C^k([-K,K])$ with an element $(m,\mu)$ of the dual space $\mathbb R^k\times \mathcal M([-K,K])$ reads as
\begin{align*}
\langle Y_t, (m,\mu)\rangle 
&=
\sum_{j=0}^{k-1} m_j \partial_x^j|_{x=0} Y^x_t + \int_{-K}^K \partial_x^kY^x_t \mu(dx)
\\&=
\sum_{j=0}^{k-1} m_j \int_0^t (s-t)^j dW_s + \int_{-K}^K \int_0^t (s-t)^k e^{-(t-s)x} dW_s \mu(dx).
\end{align*}
By the stochastic Fubini theorem (\autoref{thm:fubini}), the order of the integrals in the last expression can be exchanged. The assumptions of \autoref{thm:fubini} are satisfied because $\mu$ is a finite measure and the integrand is bounded. This shows that $\langle Y_t,(m,\mu)\rangle$ is Gaussian. As $(m,\mu)$ was arbitrary, $Y_t$ is Gaussian on $C^k([-K,K])$, for each fixed $t$. A similar argument shows that $Z_t$ is Gaussian on the same space. 
\end{proof}

\section{Fractional Brownian motion as a functional of a Markov process}\label{sec:fbm}

The goal of this section is to obtain a Markovian representation of fractional Brownian motion (fBM) in terms of $(Y,Z)$. We refer to \autoref{rem:fbm:lit} below for a comparison to earlier representations. Our starting point is the definition of fBM by Mandelbrot and Van Ness \cite{mandelbrot1968fractional}.

\begin{definition}[fBM]\label{def:fbm}
Fractional Brownian motion $W^H$ with initial value $w^H_0\in\mathbb R$ and Hurst index $H \in (0,1)$ is defined for each $t\geq0$ as 
\begin{equation}\label{equ:fbm:def2}\begin{aligned}
W^\mathrm{H}_t=w^H_0&+\frac{1}{\Gamma(H+\frac{1}{2})}\int_{-\infty}^0\left((t-s)^{H-\frac{1}{2}}-(-s)^{H-\frac{1}{2}}\right)dW_s
+\frac{1}{\Gamma(H+\frac{1}{2})}\int_{0}^t(t-s)^{H-\frac{1}{2}}dW_s,
\end{aligned}\end{equation}
where $W=(W_t)_{t\in\mathbb R}$ is two-sided Brownian motion as defined in \autoref{sec:setup}. 
\end{definition}

\subsection{\texorpdfstring{Markovian representation of fBM on $L^1$-spaces}{Markovian representation of fBM on L1-spaces}}\label{sec:fbm:l1}

\begin{definition}[Markovian representation]\label{def:fbmyz}
Let $(Y,Z)$ be the process in \autoref{def:ou} with initial value 
\begin{align*}
Y^x_0&=\int_{-\infty}^0 e^{sx}dW_s,
&
Z^x_0=-\int_{-\infty}^0 se^{xs}dW_s.
\end{align*}
Furthermore, let $\mu,\nu$ be the sigma-finite measures on $(0,\infty)$ defined as follows: for $H<1/2$, 
\begin{align*}
\mu(dx)&=\frac{dx}{x^{\frac12+H}\Gamma(H+\frac12)\Gamma(\frac12-H)} ,
&
\nu(dx)&=\frac{dx}{x^H\Gamma(\frac12+H)\Gamma(\frac32-H)}.
\end{align*}
and for $H>1/2$,
\begin{align*}
\mu(dx)&=\frac{dx}{x^H\Gamma(H+\frac12)\Gamma(\frac12-H)} ,
&
\nu(dx)&=\frac{dx}{x^{H-\frac12}\Gamma(\frac12+H)\Gamma(\frac32-H)}.
\end{align*}
\end{definition}

\begin{remark}
The constants in the definition of $\mu$ and $\nu$ in \autoref{def:ou} are not unique: if $H<1/2$ (or $H>1/2$, resp.), then $\nu$ (or $\mu$, resp.) may be multiplied by any positive constant without affecting the validity of the statements below.
\end{remark}

\begin{remark}\label{rem:fbm_infty} 
The measures $\mu,\nu$ in the definition above satisfy \autoref{ass:integrability1}, but not \autoref{ass:integrability_inf}. It follows by \autoref{thm:paths} that $(Y,Z)$ has continuous paths in $L^1(\mu_\infty)\times L^1(\nu_\infty)$ with $(\mu_\infty,\nu_\infty)$ as in \autoref{equ:mu_nu_inf}, but not necessarily in $L^1(\mu)\times L^1(\nu)$. Nevertheless, $(Y-Y_0,Z-Z_0)$ has continuous paths in $L^1(\mu)\times L^1(\nu)$, as shown in the proof of \autoref{thm:fbm_markov}.
\end{remark}

\begin{theorem}[Markovian representation]\label{thm:fbm_markov}
Under the specifications of \autoref{def:fbm,def:fbmyz}, fBM has the representation
\[
W^H_t =\left\{
\begin{aligned}
&w^H_0+\int_0^\infty (Y^x_t-Y^x_0) \mu(dx),&\text{if $H<\frac12$,}\\
&w^H_0+\int_0^\infty (Z^x_t-Z^x_0) \nu(dx),&\text{if $H>\frac12$,}
\end{aligned}
\right.
\]
where $(Y-Y_0,Z-Z_0)$ is a continuous process in $L^1(\mu)\times L^1(\nu)$.
\end{theorem}

\begin{remark}
\label{rem:fbm:lit}
Markovian representations of the integral $\int_{0}^t$ in \autoref{def:fbm} were found by Carmona and Coutin \cite{carmona1998fractional} for $H<1/2$ and by Carmona, Coutin, and Montseny \cite{carmona2000approximation} for $H>1/2$. Muravlev \cite{muravlev2011representation} incorporated also the integral $\int_{-\infty}^0$ in his representation and interpreted it as a random initial value. Moreover, in contrast to \cite{carmona2000approximation}, his representation is time-homogeneous also in the case $H>1/2$. 

The idea for deriving these representations from \eqref{equ:fbm:def2} in the case $H<1/2$ is to express the function $t\mapsto t^{H-1/2}$ in \eqref{equ:fbm:def2} as a Laplace transform,
\begin{equation*}
t^{H-1/2} \propto \int_0^\infty e^{-tx} x^{-1/2-H}dx,
\end{equation*}
and to apply the stochastic Fubini theorem. Note that for $H<1/2$ the process $Z$ is not used. 

For $H>1/2$ the function $t^{H-1/2}$ is not a Laplace transform of any measure, but one has the following two integral expressions,
\begin{equation*}
t^{H-1/2} \propto \int_0^\infty (e^{-tx}-1)x^{-1/2-H}dx, 
\qquad 
t^{H-1/2} \propto t \int_0^\infty e^{-tx} x^{1/2-H} dx.
\end{equation*}
The first integral expression leads to Muravlev's representation of fBM, which in our notation and with our choice of constants reads as
\begin{equation*}
W^H_t = w^H_0 + \int_0^\infty (Y^x_t-Y^x_0-W_t)\frac{dx}{x^{1/2+H}\Gamma(H+\frac12)\Gamma(\frac12-H)},
\end{equation*}
and the second integral expression leads to our representation of fBM.

While Muravlev's representation has the advantage of being more parsimonious, ours can be written as a linear functional of a Markov process on a state space of $L^1$ functions. It is not obvious how this can be done using Muravlev's representation; the problem is that $Y_t-Y_0$ is not integrable with respect to $x^{-1/2-H}dx$ and that $Y_t-W_t$ is not Markov. 
\end{remark}

\begin{proof}[Proof of \autoref{thm:fbm_markov} for $H<\frac12$]
The function $\tau\mapsto\tau^{H-\frac12}/\Gamma(H+\frac12)$ on $(0,\infty)$ appearing in the definition of $W^H$ is the Laplace transform of $\mu$, i.e., for each $\tau>0$ and $H<\frac12$
\[
\mathcal L(\mu)(\tau)=\int_0^\infty e^{-\tau x}\mu(dx)=\frac{\tau^{H-\frac12}}{\Gamma(H+\frac12)}.
\]
Therefore, 
\[
\begin{aligned}
W^H_t&=w^H_0+\int_{-\infty}^0 \int_0^\infty \left(e^{-x(t-s)}-e^{-x(-s)}\right)\mu(dx)dW_s
+\int_0^t \int_0^\infty e^{-x(t-s)} \mu(dx)dW_s.
\end{aligned}
\]
By the stochastic Fubini's theorem \ref{thm:fubini}, 
\[
\begin{aligned}
W^H_t&=w^H_0+\int_0^\infty \int_{-\infty}^0\left(e^{-x(t-s)}-e^{-x(-s)}\right)dW_s\mu(dx)+\int_0^\infty \int_0^t e^{-x(t-s)} dW_s\mu(dx).
\end{aligned}
\]
Condition \eqref{equ:fubini2} of Fubini's theorem is satisfied because
\[
\begin{aligned}
&\int_0^\infty \sqrt{\int_{-\infty}^0 \left(e^{-x(t-s)}-e^{-x(-s)}\right)^2ds}\mu(dx)=\int_0^\infty \frac{1-e^{-t x}}{\sqrt{2x}}\mu(dx)\leq\int_0^\infty \sqrt{\frac{1-e^{-t x}}{x}}\mu(dx)<\infty,\\
&\int_0^\infty \sqrt{\int_0^t e^{-2x(t-s)}}\mu(dx)\leq\int_0^\infty \sqrt{\frac{1-e^{-2tx}}{x}}\mu(dx)<\infty,
\end{aligned}
\]
where we use $1-e^{-tx}\leq\sqrt{1-e^{-tx}}$ and  \autoref{equ:elem_int5}. By the definition of $Y^x_t$,
\begin{align*}
W^H_t&=w^H_0+\int_0^\infty \left(e^{-xt}-1\right)Y^x_0 \mu(dx)+\int_0^\infty\int_0^te^{-x(t-s)}dW_s\mu(dx)
=w^H_0+\int_0^\infty (Y^x_t-Y^x_0)\mu(dx).
\end{align*}
The expressions 
\[
\left(e^{-xt}-1\right)Y^x_0\quad\text{and}\quad\int_0^te^{-x(t-s)}dW_s
\]
define continuous $L^1(\mu)$-valued processes: the first expression has majorant $(1\vee t)(1\wedge x)Y^x_0$ in $L^1(\mu)$, which allows one to apply the dominated convergence theorem, and the second expression is treated in \autoref{thm:paths}.
\end{proof}

\begin{proof}[Proof of \autoref{thm:fbm_markov} for $H>\frac12$]
As the function $\tau\mapsto\tau^{H-\frac32}/\Gamma(H+\frac12)$ is the Laplace transform of the measure $\nu$, the relation
\[
\tau\mathcal L(\nu)(\tau)=\tau\int_0^\infty e^{-x\tau}\nu(dx)=\frac{\tau^{H-\frac12}}{\Gamma(H+\frac12)}
\]
holds for each $\tau>0$ and $H\in (\frac12,1)$. Therefore, 
\[
\begin{aligned}
W^H_t&=w^H_0+\int_{-\infty}^0 \int_0^\infty \left((t-s)e^{-x(t-s)}+se^{xs}\right) \nu(dx)dW_s
+\int_0^t \int_0^\infty (t-s)e^{-x(t-s)} \nu(dx)dW_s.
\end{aligned}
\]
By the stochastic Fubini theorem \ref{thm:fubini}, 
\begin{equation}\label{equ:bh_after_fubini}
\begin{aligned}
W^H_t&=w^H_0+\int_0^\infty \int_{-\infty}^0\left((t-s)e^{-x(t-s)}+se^{xs}\right)dW_s\nu(dx)+\int_0^\infty \int_0^t (t-s)e^{-x(t-s)} dW_s\nu(dx).
\end{aligned}
\end{equation}
Condition \eqref{equ:fubini2} of Fubini's theorem is satisfied because
\[
\begin{aligned}
&\int_0^\infty \sqrt{\int_{-\infty}^0 \left((t-s)e^{-x(t-s)}+se^{xs}\right)^2ds}\nu(dx)
=\int_0^\infty \sqrt{\frac{1-2 e^{-t x} (t x+1)+2 t x e^{-2 t x} (t x+1)+e^{-2 t x}}{4 x^3}}\nu(dx)\\&\quad\leq\int_0^{1/t} \sqrt{\frac{t^2}{6 x}}\nu(dx)+\int_{1/t}^\infty \sqrt{\frac{2}{x^3}}\nu(dx)\leq\sqrt{2}(t\vee 1)\int_0^\infty (x^{-\frac12}\wedge x^{-\frac32})\nu(dx)<\infty,\\
&\int_0^\infty \sqrt{\int_0^t (t-s)^2 e^{-2x(t-s)}}\nu(dx)=\int_0^\infty \sqrt{\frac{1-e^{-2tx}\left(1+2tx+2t^2x^2\right)}{4x^3}}\nu(dx)<\infty,
\end{aligned}
\]
where we used Equations \eqref{equ:elem_int6} and \eqref{equ:elem_int7}. Using the definition of $(Y^x,Z^x)$ in \autoref{equ:ou}, \autoref{equ:bh_after_fubini} can be expressed as
\begin{equation*}
\begin{aligned}
W^H_t&=w^H_0+\int_0^\infty \int_{-\infty}^0e^{xs}\left(te^{-xt}+s(1-e^{-xt})\right)dW_s \nu(dx)+\int_0^\infty \int_0^t (t-s)e^{-x(t-s)} dW_s\nu(dx)
\\&=w^H_0+\int_0^\infty \left(te^{-xt} \int_{-\infty}^0 e^{xs} dW_s+(1-e^{-xt})\int_{-\infty}^0 se^{xs} dW_s \right)\nu(dx)\\&\quad
+\int_0^\infty \left(Z^x_t - Z^x_0 e^{-xt}-Y^x_0te^{-xt}\right)\nu(dx)\\&=
w^H_0+\int_0^\infty \left( Z^x_t -Z^x_0 \right)\nu(dx).
\end{aligned}\end{equation*}
By \autoref{lem:representation_ou}, $Z^x_t-Z^x_0$ can be written as the sum of the following expressions:
\begin{align*}
&Z^x_0 (e^{-tx}-1), 
&& Y^x_0te^{-tx},
&& \int_0^t(t-s)e^{-(t-s)x}dW_s.
\end{align*}
All three expressions define continuous $L^1(\nu)$-valued processes: the first and second expression have $|Z^x_0|(1\vee t)(1\wedge x)$ and $|Y^x_0|(1\vee t)(1\wedge x^{-1})$ as majorants in $L^1(\nu)$, which allows one to apply the dominated convergence theorem, and the third expression is treated in \autoref{thm:paths}.
\end{proof}

\begin{remark}
The representation in \autoref{thm:fbm_markov} lends itself to numerical implementation. Indeed, the integrals can be approximated by finite sums as described in \cite{carmona2000approximation}. Alternatively, aiming for a more parsimonious representation, one has in the case $H>1/2$
\[
W^H_t =w^H_0-\int_0^\infty \partial_x(Y^x_t-Y^x_0) \nu(dx).
\]
This follows from the following deterministic relationship between $Y$ and $Z$ (c.f. \autoref{thm:smoothness})
\begin{equation}\label{equ:nonminimality_of_z}
Z^x_t = -\partial_x Y^x_t + (\partial_x Y^x_0 + Z^x_0)e^{-tx}, \quad t\geq 0.
\end{equation}
\end{remark}

\begin{remark}
The case $H=1/2$ fits into the framework of \autoref{thm:fbm_markov} with $\mu$ equal to the Dirac measure. Indeed, the process $(Y^0_t-Y^0_0)_{t\geq 0}$ is Brownian motion, as can be seen from the definition of $Y$. Moreover, the choice of $\mu$ as a Dirac measure is in line with the proof of \autoref{thm:fbm_markov} where $\mu$ is defined as the inverse Laplace transform of the integrand in \autoref{def:fbm}. Note that the representing Markov process $Y_t \in L^1(\mu)$ is one-dimensional and can be identified with Brownian motion. 
\end{remark}

\begin{remark}\label{rem:fbm:generalizations}
The arguments of \autoref{thm:fbm_markov} yield Markovian representations of all fractional processes of the form
\begin{equation*}
\int_{-\infty}^0 \big(k(t-s)-k(-s)\big)dW_s+\int_0^t k(t-s)dW_s, 
\end{equation*}
where $k(t)=t^n \mathcal L(\mu)(t)$, $n \in \{0,1\}$, $\mu$ is a sigma-finite measure on $\mathbb R_+$, and $\mathcal L$ is the Laplace transform. There are many examples in the theory of (semi-)stationary processes, including power kernels $k(t)=t^\alpha (1+t)^{-\gamma-\alpha}$ and Gamma kernels $k(t)=t^\alpha \exp(-\lambda t)$.
More generally, $W$ could be replaced by a L\'evy process, which would lead to Markovian representations of fractional L\'evy processes and stable L\'evy motions \cite{samorodnitsky1994stable, marquardt2006fractional, basse2009levy, engelke2013unifying, kluppelberg2015generalized, bennedsen2017hybrid}. 
\end{remark}

\subsection{Filtrations}\label{sec:filtrations}

The filtration generated by $W^H$ is essentially the same as the one generated by $(Y,Z)$, as shown in the following lemmas. Therefore, the law of fractional Brownian motion after a stopping time can be characterized using the strong Markov property of $(Y,Z)$. This is important for understanding the existence of arbitrage opportunities in models with fractional price processes (see. e.g. the stickiness property in \cite{guasoni2006no,czichowsky2017portfolio} and the notion of arbitrage times in \cite{peyre2015no}). 

\begin{lemma}[Filtrations]
Let $H<1/2$. Then the completed filtrations generated by the processes $W-W_0$, $W^H-W^H_0$, and $Y-Y_0$ are equal. The same statement holds for $H>1/2$ with $Y$ replaced by $Z$. 
\end{lemma}

\begin{proof}
Let $\mathcal N$ denote the $\mathbb Q$-null sets. Then the following sigma algebras are equal for each $T\geq 0$: 
\begin{equation*}\begin{alignedat}{4}
\sigma(W_t-W_0,0\leq t\leq T)\vee \mathcal N 
&= \sigma(&W^H_t&-W^H_0,\ &&0\leq t\leq T)&&\vee\mathcal N 
\\&\subseteq \sigma(&Y_t\ \,&-\ Y_0,&&0\leq t\leq T)&&\vee \mathcal N
\\&\subseteq \sigma(&W_t\;&-\;W_0,&&0\leq t\leq T)&&\vee \mathcal N.
\end{alignedat}\end{equation*}
The first equality above follows from \cite[Proposition~1]{pipiras2002deconvolution}. The proof for $H>1/2$ is similar.
\end{proof}

From a Markovian point of view, the canonical definition of fractional Brownian motion is $V^H_t=\langle Y_t,1\rangle_\mu$ or $V^H_t=\langle Z_t,1\rangle_\nu$, depending on whether $H$ is smaller or greater than $1/2$. Here the initial value $(Y_0,Z_0)$ is fixed and deterministic. Moreover, the initial value $W_0$ can be normalized to zero. Then the following lemma holds.

\begin{lemma}[Filtrations]
If $H>1/2$, then the completed filtrations generated by the processes $W$, $V^H$, and $Y$ are equal. The same statement holds for $H>1/2$ with $Y$ replaced by $Z$. 
\end{lemma}

\begin{proof}
As before, $\mathcal N$ denotes the $\mathbb Q$-null sets. Let us assume for a moment that the initial value $(Y_0,Z_0)$ is zero. Then one has for each $T\geq 0$
\begin{equation*}\begin{alignedat}{4}
\sigma(W_t,0\leq t\leq T)\vee \mathcal N &= \sigma(&V^H_t,\ &0\leq t\leq T)\vee\mathcal N 
\\&\subseteq \sigma(&Y_t\ \,,\ &0\leq t\leq T)\vee \mathcal N
\\&\subseteq \sigma(&W_t\;,\ &0\leq t\leq T)\vee \mathcal N.
\end{alignedat}\end{equation*}
The first equality above follows from \cite[Proposition~1]{pipiras2002deconvolution} applied to a Brownian path which is set to zero for all $t\leq 0$, noting that the relevant integrals are  defined pathwise. To get rid of the assumption on $(Y_0,Z_0)$, note that the process $(Y,Z)$ depends on the initial condition $(Y_0,Z_0)$ only via a deterministic function, which is $\mathcal N$-measurable. The proof for $H>1/2$ is similar.
\end{proof}

\subsection{\texorpdfstring{Markovian representation of fBM on $L^2$-spaces}{Markovian representation of fBM on L2-spaces}}

There is also an $L^2$-version of the results of \autoref{sec:fbm:l1}.

\begin{theorem}[Markovian representation]\label{thm:fbm_l2}
Let $(Y,Z)$ and $(\mu,\nu)$ be as in \autoref{def:fbmyz}, 
let 
\begin{equation*}
f\colon(0,\infty)\to(0,\infty), \qquad x\mapsto 1\wedge x^{-1/2},
\end{equation*}
let $\tilde \mu=\mu f^{-1}$, and let $\tilde \nu=\nu f^{-1}$.
Then fBM has representation
\[
W^H_t =\left\{
\begin{aligned}
&w^H_0+\int_0^\infty (Y^x_t-Y^x_0)f(x) \tilde\mu(dx),&\text{if $H<\frac12$,}\\
&w^H_0+\int_0^\infty (Z^x_t-Z^x_0)f(x) \tilde\nu(dx),&\text{if $H>\frac12$,}
\end{aligned}
\right.
\]
where $(Y-Y_0,Z-Z_0)$ is a continuous $L^2(\tilde\mu)\times L^2(\tilde\nu)$-valued process and $f \in L^2(\tilde\mu)\cap L^2(\tilde\nu)$. 
\end{theorem}

This can be shown along the lines of the proof of \autoref{thm:fbm_markov}.

\begin{remark}
The measures $\tilde\mu$ and $\tilde\nu$ above satisfy \autoref{ass:l2,ass:integrability_l2_stronger}, but $(Y_\infty,Z_\infty)$ does not take values in $L^2(\tilde\mu)\times L^2(\tilde\nu)$ (c.f. \autoref{rem:fbm_infty}). Nevertheless, the process $(Y-Y_0,Z-Z_0)$ does, as stated in \autoref{thm:fbm_l2}.
\end{remark}

\section{Applications to interest rate modeling}\label{sec:applications}

\subsection{Fractional short rate models}\label{sec:fractional_sr}

Our prototypical example is
\begin{equation*}
r_t = \ell + \lambda V^H_t, \qquad V^H_t = \frac{1}{\Gamma(H+\frac12)}\int_0^t (t-s)^{H-1/2} dW_s,
\end{equation*} 
where $r_t$ is the short rate, $\ell,\lambda \in \mathbb R$, and $V^H_t$ is Volterra fractional Brownian motion. For $H>1/2$ this is consistent with the empirical observation that the short rate has long-range dependence \cite{backus1993emp}. In contrast to \cite{ohashi2009ftsm} and \cite{biagini2013ftsm}, discounted zero-coupon bond prices are martingales. As in any Ornstein--Uhlenbeck type model the short rate in our model may become negative, which one may find reasonable or not. 

Taking this example as a starting point, we will look at slightly more general models which are defined in terms of the processes $Y$ and $Z$. To this aim we fix measures $\mu, \nu$ satisfying the following strengthened version of \autoref{ass:integrability1}.

\begin{assumption}\label{ass:integrability_shortrate}
$\mu$ and $\nu$ are sigma-finite measures on $(0,\infty)$. The measure $\nu$ has a density $p$ with respect to $\mu$, and there exists $\beta\in(0,2)$ such that for each $t>0$,
\begin{align*}
&\int_0^\infty (1\wedge x^{-\frac12})\mu(dx)<\infty,
&
&\int_0^\infty (1\wedge x^{-\frac32}) \nu(dx)<\infty,
&
&\sup_{x\in(0,\infty)}p(x)(1\wedge x^{-\beta})<\infty.
\end{align*}
\end{assumption}

Moreover, we fix $(u,v) \in L^\infty(\mu)\times L^\infty(\nu)$, $\ell \in \mathbb R$, and an initial value $(Y_0,Z_0)\in L^1(\mu)\times L^1(\nu)$ for the process $(Y,Z)$ defined in \autoref{sec:ou}. Often times, either $u$ or $v$ will be set  to zero, unless one is interested in mixing processes with long- and short-range dependence. Given these model parameters, we define for each $\tau,T \in \mathbb R_+$ and $t \in [0,T]$ the short rate $r_t$, bank account $B_t$, ZCB prices $P(t,T)$, and forward rates $h(t)(\tau)$ as
\begin{equation}\label{equ:fractional_sr}
r_t=\ell+\left\langle Y_t,u\right\rangle_\mu+\left\langle Z_t,v\right\rangle_\nu,
\qquad
B_t=\exp\left(\int_0^t r_s ds\right),
\qquad
P(t,T)=\mathbb E\left[\frac{B_t}{B_T}\middle|\mathcal F_t\right]=e^{-\int_0^{T-t}h(t)(\tau)d\tau}.
\end{equation}

\begin{theorem}[Bond prices and forward rates]\label{thm:fractional_sr_zcb}
In the fractional short rate model \eqref{equ:fractional_sr}, ZCB prices and forward rates are given by
\begin{align*}
P(t,T)&=e^{-\ell(T-t)+\Phi_0(T-t,u,v)+\left\langle Y_t,\Phi_1(T-t,u,v)\right\rangle_\mu+\left\langle Z_t,\Phi_2(T-t,u,v)\right\rangle_\nu},
&&0\leq t\leq T,\\
h(t)(\tau)&=\ell-\partial_\tau\Phi_0(\tau,u,v)-\left\langle Y_t,\partial_\tau\Phi_1(\tau,u,v)\right\rangle_\mu-\left\langle Z_t,\partial_\tau\Phi_2(\tau,u,v)\right\rangle_\nu,
&&t,\tau\geq0,
\end{align*}
where for each $\tau\geq0$ and $x\in(0,\infty)$
\[
\begin{aligned}
\Phi_0(\tau,u,v)&=\frac{1}{2}\int_0^{\tau}\langle \Phi_1(s,u,v),1\rangle_\mu^2ds,\\
\Phi_1(\tau,u,v)(x)&=\frac{e^{-\tau x}-1}{x}u(x)+ \left(\frac{e^{-\tau x}-1}{x^2}+\frac{\tau}{x}e^{-\tau x}\right)p(x)v(x),\\
\Phi_2(\tau,u,v)(x)&=\frac{e^{-\tau x}-1}{x} v(x).
\end{aligned}
\]
\end{theorem}

\begin{proof}
\autoref{lem:representation_int_inf_dim_ou} implies that the random variable $\int_{t}^{T}\left(\left\langle Y_s,u\right\rangle_\mu+\left\langle Z_s,v\right\rangle_\nu\right)ds$ is Gaussian, given $\mathcal F_t$, with mean 
\[
-\left\langle Y_t,\Phi_1(T-t,u,v)\right\rangle_\mu-\left\langle Z_t,\Phi_2(T-t,u,v)\right\rangle_\nu
\]
and variance $2\Phi_0(T-t,u,v)$. Thus, the formula for ZCB prices follows from the formula of the moment generating function of the normal distribution. The expression for the forward rates follows by differentiation with respect to the time to maturity.
\end{proof}

\begin{remark}
The functions $\Phi_0,\Phi_1,\Phi_2$ are the unique solution of the Riccati equations
\begin{equation}
\begin{aligned}
\partial_\tau\Phi_0(\tau,u,v)&= R_0\big(\Phi_1(\tau,u,v),\Phi_2(\tau,u,v)\big),&\Phi_0(0,u,v)&=0,\\
\partial_\tau\Phi_1(\tau,u,v) &= R_1\big(\Phi_1(\tau,u,v),\Phi_2(\tau,u,v)\big)-u,&\Phi_1(0,u,v)&=0,\\
\partial_\tau\Phi_2(\tau,u,v) &= R_2\big(\Phi_1(\tau,u,v),\Phi_2(\tau,u,v)\big)-v,&\Phi_2(0,u,v)&=0,
\end{aligned}
\end{equation}
with $R_0,R_1,R_2$ as in \autoref{lem:riccati}. Here, solutions are defined in analogy to \autoref{def:riccati} and \autoref{lem:riccati}.
\end{remark}

\begin{theorem}[HJM equation]\label{thm:fractional_sr_hjm}
In the fractional short rate model \eqref{equ:fractional_sr} the bank account $(B_t)_{t\geq 0}$, bond prices $(P(t,T))_{0\leq t\leq T}$ and forward rates $(h(t)(\tau))_{t\geq0}$ are semimartingales for each fixed $T,\tau>0$. 
The forward rate process $h=(h(t)(\cdot))_{t\geq0}$ is a solution of the HJM equation
\begin{equation}\label{equ:hjm}
dh(t)=\left(\mathcal Ah(t)+\mu^{\mathrm{HJM}}\right)dt+\sigma^{\mathrm{HJM}}dW_t,
\end{equation}
where $\mathcal A$ denotes differentiation with respect to time to maturity $\tau$ and $\mu^\mathrm{HJM},\sigma^\mathrm{HJM}$ are measurable functions on $(0,\infty)$ given by
\[ 
\begin{aligned}
\mu^{\mathrm{HJM}}(\tau)&=\partial_\tau^2\Phi_0(\tau,u,v),&
\sigma^{\mathrm{HJM}}(\tau)&=-\langle\partial_\tau\Phi_1(\tau,u,v),1\rangle_\mu.
\end{aligned}
\]
\end{theorem}

\begin{remark}
The HJM drift condition is satisfied because
\[
\mu^{\mathrm{HJM}}=\partial_\tau^2\Phi_0(\tau,u,v)
=\langle\partial_\tau\Phi_1(\tau,u,v),1\rangle_\mu \langle\Phi_1(\tau,u,v),1\rangle_\mu\\
=\sigma^{\mathrm{HJM}}(\tau)\int_0^\tau\sigma^{\mathrm{HJM}}(s)ds.
\]
\end{remark}

\begin{proof}[Proof of \autoref{thm:fractional_sr_hjm}]
The semimartingale property of prices and forward rates follows from \autoref{lem:semimartingale_Psi,lem:semimartingale_Psi_prime}, which are based on \autoref{thm:semimartingale}. The semimartingale decomposition of $h(\cdot)(\tau)$ is obtained by collecting the terms in \autoref{equ:semimartingale_decompositions}:
\begin{equation*}\begin{aligned}
dh(t)(\tau)
&=
-d\left\langle Y_t,\partial_\tau\Phi_1(\tau,u,v)\right\rangle_\mu
-d\left\langle Z_t,\partial_\tau\Phi_2(\tau,u,v)\right\rangle_\nu 
\\
&=\left(\langle Y_t,x\partial_\tau\Phi_1(\tau,u,v)-\partial_\tau\Phi_2(\tau,u,v)p\rangle_\mu+\langle Z_t,x\partial_\tau\Phi_2(\tau,u,v)\rangle_\nu\right)dt
-\langle\partial_\tau\Phi_1(\tau,u,v),1\rangle_\mu dW_t.
\end{aligned}\end{equation*}
Note that by abuse of notation, we wrote $x\partial_\tau\Psi_{i}(\tau,u,v)$ to designate the function $x\mapsto \partial_\tau\Psi_{i}(\tau,u,v)(x)$ for $i=1,2$. The second derivatives of $\Psi_i$ are
\begin{equation*}\begin{aligned}
\partial_\tau^2 \Phi_1(\tau,u,v)
&=
-x\partial_\tau\Phi_1(\tau,u,v)+\partial_\tau\Phi_2(\tau,u,v)p, 
\\
\partial_\tau^2 \Phi_2(\tau,u,v)
&=
-x\partial_\tau\Phi_2(\tau,u,v).
\end{aligned}\end{equation*}
Therefore, we have for all $t\geq0$ and $\tau>0$
\begin{equation*}
\mathcal Ah(t)(\tau)=-\partial^2_\tau\Phi_0(\tau,u,v)+\left\langle Y_t,x\partial_\tau\Phi_1(\tau,u,v)-\partial_\tau\Phi_2(\tau,u,v)p\right\rangle_\mu
+\left\langle Z_t,x\partial_\tau\Phi_2(\tau,u,v)\right\rangle_\nu.
\end{equation*}
It follows that
\[
dh(t)(\tau)=\left(\mathcal Ah(t)(\tau)+\partial_\tau^2\Phi_0(\tau,u,v)\right)dt-\langle\partial_\tau\Phi_1(\tau,u,v),1\rangle_\mu dW_t,
\]
which allows one to identify $\mu^\mathrm{HJM}$ and $\sigma^\mathrm{HJM}$.
\end{proof}

\begin{corollary}[Covariations]\label{cor:fractional_sr_mu_nu}
For each $\tau_1,\tau_2>0$ the following relation holds:
\begin{equation}\label{equ:fwd_qv_fractional_sr}
d[h(\cdot)(\tau_1),h(\cdot)(\tau_2)]_t=\langle\partial_\tau\Phi_1(\tau_1,u,v),1\rangle_\mu \langle\partial_\tau\Phi_1(\tau_2,u,v),1\rangle_\mu dt.
\end{equation}
\end{corollary}

The above formula for realized covariations can be used to calibrate the model to historical time series. Another way is to calibrate to option prices, including caps and floors, which can be done efficiently thanks to the following closed-form expression of vanilla ZCB option prices. 

\begin{theorem}[Black--Scholes formula]\label{thm:fractional_sr_v}
The prices of call and put options on zero coupon bonds are given by the following version of the Black--Scholes formula,
\begin{equation}\label{equ:black_scholes1}
\begin{aligned}
\mathbb E\left[
\frac{B_0}{B_T}\big(P(T,S)-K\big)^+\middle|\mathcal F_t\right]
&=
P(t,S)\Phi_0^{\mathrm{Gauss}}(d_1)-KP(t,T)\Phi_0^{\mathrm{Gauss}}(d_2),\\
\mathbb E\left[
\frac{B_0}{B_T}\big(K-P(T,S)\big)^+\middle|\mathcal F_t\right]
&=
KP(t,T)\Phi_0^{\mathrm{Gauss}}(-d_2)-P(t,S)\Phi_0^{\mathrm{Gauss}}(-d_1),
\end{aligned}
\end{equation}
where $\Phi_0^{\mathrm{Gauss}}$ is the standard Gaussian cumulative distribution function,
\begin{equation}\label{equ:black_scholes2}
d_{1,2}=\frac{\log\left(\frac{P(t,S)}{KP(t,T)}\right)\pm\frac12 \int_t^T\big(v(\cdot,S)-v(\cdot,T)\big)^2ds}{\sqrt{\int_t^T\big(v(\cdot,S)-v(\cdot,T)\big)^2ds}},
\end{equation}
and where $v(t,T)=\langle\Phi_1(T-t,u,v),1\rangle_\mu$.
\end{theorem}

\begin{proof}
In \autoref{lem:semimartingale_Psi} we verified that the expressions $\left\langle Y,\Phi_1(T-\cdot,u,v)\right\rangle_\mu$ and $\left\langle Z,\Phi_2(T-\cdot,u,v)\right\rangle_\nu$ are semimartingales. Their semimartingale decompositions are given by \autoref{equ:semimartingale_decompositions}:
\[
\begin{aligned}
d\left\langle Y_t,\Phi_1(T-t,u,v)\right\rangle_\mu&=\int_0^\infty\left(u(x)-\frac{e^{-(T-t) x}-1}{x}p(x)v(x)\right)Y^x_t\mu(dx)dt
+\langle\Phi_1(T-t,u,v),1\rangle_\mu dW_t,\\
d\left\langle Z_t,\Phi_2(T-t,u,v)\right\rangle_\nu&=\Big(\left\langle v,Z_t\right\rangle_\nu+\left\langle \Phi_2(T-t,u,v),Y_t\right\rangle_\nu\Big)dt.
\end{aligned}
\]
Let $\xi(t,T)$ be the density of the $T$-forward measure $\mathbb Q^T$, 
\begin{equation*}
\xi(t,T)=\left.\frac{d\mathbb Q^T}{d\mathbb Q}\right|_{\mathcal F_t}=\mathbb E\left[\frac{B_0}{P(0,T)B_T}\middle|\mathcal{F}_t\right]=\frac{B^{-1}_tP(t,T)}{B^{-1}_0P(0,T)}.
\end{equation*}
By the formula for bond prices in \autoref{thm:fractional_sr_zcb}, $\log(\xi(t,T))$ satisfies
\[
\begin{aligned}
d(\log\xi(t,T))&=\left(-\left\langle Y_t,u\right\rangle_\mu-\left\langle Z_t,v\right\rangle_\nu-\partial_\tau\Phi_0(T-t,u,v)\right)dt\\&\quad+d\left\langle Y_t,\Phi_1(T-t,u,v)\right\rangle_{\mu}+d\left\langle Z_t,\Phi_2(T-t,u,v)\right\rangle_{\nu}.
\end{aligned}
\]
Applying It\=o's formula and canceling out terms yields
\[
\begin{aligned}
d\xi(t,T)&=\xi(t,T)\left(d(\log\xi(t,T))+\frac{1}{2}d\left[\log\xi(\cdot,T)\right]_t\right)=\xi(t,T)\left\langle\Phi_1(T-t,u,v),1\right\rangle_\mu dW_t,
\end{aligned}
\]
which implies that $\xi$ is a stochastic exponential of the form 
\begin{equation}\label{equ:xi_exp}
\xi(t,T)=\mathcal E\left(\int_0^\cdot v(s,T)dW_s\right)_t
\end{equation}
with $v(t,T)=\left\langle\Phi_1(T-t,u,v),1\right\rangle_\mu$. Then, for any $S,T>0$ the process $W^T=W-\int_0^\cdot v(s,T)ds$ is $\mathbb Q^T$-Brownian motion, and 
\[
\frac{P(t,S)}{P(t,T)} 
=\frac{P(0,S)}{P(0,T)}\mathcal E\left(\int_0^\cdot\big(v(s,S)-v(s,T)\big)dW_s^T\right)_t,\quad t\in[0,S\wedge T],
\]
is a log-normal $\mathbb Q^T$-martingale. Then the Black--Scholes formulas hold by \cite[Proposition~7.2]{filipovic2009term}.
\end{proof}

\subsection{Fractional bank account models}\label{sec:fractional_ba}

A variant of the above model is to let the bank account be a fractional process; for example, 
\begin{equation*}
B_t = \exp(\ell t + \lambda V^H_t), \qquad V^H_t = \frac{1}{\Gamma(H+\frac12)}\int_0^t (t-s)^{H-1/2} dW_s,
\end{equation*}
where $B_t$ is the bank account, $\ell,\lambda \in \mathbb R$, and $V^H_t$ is Volterra fractional Brownian motion. This defines an arbitrage-free model where the short rate does not exist and the bank account and bond prices are not semimartingales. This does not lead to arbitrage; we will see that discounted bond prices are martingales. 

We only provide a summary of the results because we do not have an economic motivation for the model. Indeed most financial models assume that the bank account has finite variation. We would like to point out, however, that from a general no-arbitrage point of view there is no reason to assume that the bank account is a semimartingale, or even to assume that a bank account process exists at all \cite{klein2013rollovers}. What we like about the fractional bank account model is the way it demonstrates the connection between general no-arbitrage theory and semimartingality of bond prices, discounted bond prices, and forward rates. 

\begin{theorem}\label{thm:fractional_ba_summary}
Let the bank account process by given by $B_t=e^{\ell t+\left\langle Y_t,u\right\rangle_\mu+\left\langle Z_t,v\right\rangle_\nu}$, where $\mu$ and $\nu$ satisfy \autoref{ass:integrability1}, $(u,v) \in L^\infty(\mu)\times L^\infty(\nu)$, $\ell \in \mathbb R$, and $(Y,Z)$ is as in \autoref{def:ou} with $(Y_0,Z_0)\in L^1(\mu)\times L^1(\nu)$. Then the ZCB prices $(P(t,T))_{t\geq 0}$ are well-defined, but are in general not semimartingales. However, the discounted ZCB prices $(B_t^{-1}P(t,T))_{t\geq 0}$ are martingales, and the forward rates $(h(t)(\tau))_{t\geq0}$ are semimartingales, which satisfy a HJM equation. The Black--Scholes formulas \eqref{equ:black_scholes1}--\eqref{equ:black_scholes2} hold with $v(t,T)=\left\langle\phi_1(\tau,-u,-v),1\right\rangle_\mu$, where $\phi_1$ is given by \autoref{thm:affine}.
\end{theorem}

\section{Fractional Stein \& Stein model}\label{sec:stein}

In this section we generalize an affine stochastic volatility model by Stein and Stein \cite{stein1991} to fractional volatility. In the original model, the volatility process is a single OU process, whereas in our model, it is a fractional process, i.e., a superposition of infinitely many OU processes. Our main example is of the form
\begin{equation*}
dS_t = S_t \sigma V^H_t d\widetilde W_t, \qquad V^H_t = \frac{1}{\Gamma(H+\frac12)}\int_0^\infty (t-s)^{H-1/2}dW_s, 
\end{equation*}
where $S$ is the asset price, $\sigma \in \mathbb R$, $W$ and $\widetilde W$ are Brownian motions, and $V^H_t$ is Volterra Brownian motion of Hurst index $H<1/2$. The restriction to $H<1/2$ is in accordance with well-documented empirical facts about realized volatility \cite{gatheral2014vola}. 

We will show that the fractional Stein and Stein model is affine. By this we mean that the log price is the first coordinate of an infinite-dimensional affine process. 

\subsection{Setup and notation}

As before we are working on a filtered probability space $(\Omega,\mathcal F,(\mathcal F_t)_{t \in \mathbb R},\mathbb Q)$ supporting a two-sided Brownian motion $W$. Let $\widetilde W$ be $(\mathcal F_t)_{t\geq 0}$-Brownian motion with correlation $d\langle W,\widetilde W\rangle_t = \rho dt$ for some $\rho \in [-1,1]$. We fix a measure $\mu$ on $(0,\infty)$ satisfying \autoref{ass:integrability1}, a function $\lambda\in L^\infty(\mu)$, and an initial value $Y_0\in L^1(\mu)$ for the process $Y$ defined in \autoref{sec:ou}. Given these model parameters, the price process $S=(S_t)_{t\geq0}$ under the risk-neutral probability measure $\mathbb Q$ is defined by the SDE 
\begin{equation*}
dS_t=S_t \langle Y_t,\lambda\rangle_\mu d\widetilde W_t.
\end{equation*}
To bring the SDE for the process $S$ into an affine form, we introduce the following spaces of simple symmetric nonnegative tensors:\footnote{All tensor products are algebraic; we do not complete the tensor products.} 
\begin{equation*}\begin{alignedat}{6}
L^1(\mu)&\otimes_s L^1(\mu)&&=\{y^{\otimes 2}\colon y \in L^1(\mu)\}&&\subset L^1(\mu)^{\otimes2} &&\subset L^1(\mu^{\otimes 2}),
\\
L^\infty(\mu)&\otimes_s L^\infty(\mu) &&= \{v^{\otimes 2}\colon v \in L^\infty(\mu)\} &&\subset L^\infty(\mu)^{\otimes2} &&\subset L^\infty(\mu^{\otimes2}).
\end{alignedat}\end{equation*}
For each $t\geq 0$ we set $\Pi_t=Y_t^{\otimes 2} \in L^1(\mu)\otimes_s L^1(\mu)$. Then the relation $\langle Y_t,\lambda\rangle_\mu^2 = \langle Y_t^{\otimes2},\lambda^{\otimes2}\rangle_{\mu^{\otimes2}}$ holds. Therefore, the log-price process $X=\log(S)$ satisfies\footnote{To be correct, $\widetilde W_t$ in \eqref{equ:stein_SDE} should be replaced by $\int_0^t\operatorname{sgn}(\langle \Pi_s,\lambda\rangle_\mu)d\widetilde W_s$, which is again a Brownian motion.}
\begin{equation}\label{equ:stein_SDE}
dX_t = -\frac12 \left\langle \Pi_t,\lambda^{\otimes2} \right\rangle_{\mu^{\otimes2}} dt + \sqrt{\left\langle \Pi_t,\lambda^{\otimes2}\right\rangle_{\mu^{\otimes2}} }d\widetilde W_t.
\end{equation}	

\subsection{\texorpdfstring{Affine structure of $\Pi$}{Affine structure of Pi}}

The following theorem characterizes $\Pi$ as an affine process with values in $L^1(\mu)\otimes_s L^1(\mu)$.

\begin{theorem}[Affine structure]\label{thm:stein:affine:special}
Let $v^{\otimes2}\in iL^\infty(\mu)\otimes_s L^\infty(\mu)$. Then, with probability one,
\[
\mathbb E\left[e^{\left\langle\Pi_T,v^{\otimes2}\right\rangle_{\mu^{\otimes2}}}\middle|\mathcal F_t\right]
=e^{\psi_0\left(T-t,v^{\otimes2}\right)+\left\langle\Pi_t,\psi_1\left(T-t,v^{\otimes2}\right)\right\rangle_{\mu^{\otimes2}}},\quad 0\leq t\leq T,
\]
where $\psi_0\left(\tau,v^{\otimes2}\right)\in\mathbb C$ and $\psi_1\left(\tau,v^{\otimes2}\right)\in L^\infty(\mu;\mathbb C)\otimes_s L^\infty(\mu;\mathbb C)$ are given by
\begin{align*}
\psi_0\left(\tau,v^{\otimes2}\right)&=-\frac{1}{2}\log\left(1-4\phi_0(\tau,v,0)\right),
&
\psi_1\left(\tau,v^{\otimes2}\right)&=\frac{\phi_1(\tau,v,0)^{\otimes2}}{1-4\phi_0(\tau,v,0)}.
\end{align*}
\end{theorem}

\begin{remark}\label{rem:sde:pi}
An immediate observation is that for each $(x,y)\in(0,\infty)^2$, the tuple $(\Pi^{x,x},\Pi^{x,y},\Pi^{y,y})$ is an affine process. This can be seen from the following SDE for $\Pi^{x,y}_t=Y^x_tY^y_t$, which follows from It\=o's rule:
\[
d\Pi_t^{x,y}=\left(1-(x+y)\Pi_t^{x,y}\right)dt+\sqrt{\Pi_t^{x,x}+2\Pi_t^{x,y}+\Pi_t^{y,y}}dW_t.
\]	
More generally, for any finite set of points $x_i$, the process $(\Pi^{x_i,x_j})_{i,j}$ is affine. \autoref{thm:stein:affine:special} generalizes this observation to infinitely many points $x^i,x^j \in (0,\infty)$. A version of \autoref{thm:stein:affine:special} with $v^{\otimes2}$ replaced by arbitrary symmetric test functions is given in \autoref{lem:stein:affine_general}. 
\end{remark}

\begin{remark}
The state space of the affine process in \autoref{thm:stein:affine:special} consists of rank one tensors. It might be possible to extend the state space to rank two tensors, but the drift condition in \cite[Proposition~4.18]{cuchiero2011affine} suggests that an extension to tensors of higher rank is not possible.
\end{remark}

\begin{proof}
By \autoref{lem:representation_inf_dim_ou} the random variable $\frac{1}{\sqrt{2\phi_0(T-t,v,0)}}\left\langle Y_T,v\right\rangle_\mu$ is Gaussian, given $\mathcal F_t$, with mean
\[
\frac{\left\langle Y_t,\phi_1(T-t,v,0)\right\rangle_\mu}{\sqrt{2\phi_0(T-t,v,0)}},
\]
and unit variance. Hence, the random variable
\[
\frac{\left\langle\Pi_T,v^{\otimes2}\right\rangle_{\mu^{\otimes2}}}{2\phi_0(T-t,v,0)}=\left(\frac{\left\langle Y_T,v\right\rangle_\mu}{\sqrt{2\phi_0(T-t,v,0)}}\right)^2,
\]
is non central $\chi^2$-distributed, given $\mathcal F_t$, with one degree of freedom and non centrality parameter
\[
\frac{\left\langle Y_t,\phi_1(T-t,v,0)\right\rangle^2_\mu}{2\phi_0(T-t,v,0)}=\frac{\left\langle\Pi_t,\phi_1(T-t,v,0)^{\otimes2}\right\rangle_{\mu^{\otimes2}}}{2\phi_0(T-t,v,0)}.
\]
The statement follows from the formula for the characteristic function of the non central $\chi^2$ distribution.
\end{proof}

The coefficient functions $(\psi_0,\psi_1)$ of \autoref{thm:stein:affine:special} are solutions of an infinite dimensional version of the Riccati ODE's in the sense of \autoref{def:riccati}.

\begin{lemma}[Riccati equations]
For any $v^{\otimes2}\in iL^\infty(\mu)\otimes_s L^\infty(\mu)$, the functions $\psi_0\left(\cdot,v^{\otimes2}\right)$ and $\psi_1\left(\cdot,v^{\otimes2}\right)$ given by \autoref{thm:stein:affine:special} solve the following system of differential equations
\[
\begin{aligned}
\partial_\tau\psi_0\left(\tau,v^{\otimes2}\right)&=F_0\big(\psi_1\left(\tau,v^{\otimes2}\right)\big),\quad\psi_0\left(0,v^{\otimes2}\right)=0,\\
\partial_\tau\psi_1\left(\tau,v^{\otimes2}\right)&= F_1\big(\psi_1\left(\tau,v^{\otimes2}\right)\big),\quad\psi_1\left(0,v^{\otimes2}\right)=v^{\otimes2},
\end{aligned}
\]
where for any $w\in L^\infty(\mu;\mathbb C)^{\otimes2}$, $F_0(w)$ is the complex number given by
\[
F_0(w)=\int_0^\infty\int_0^\infty w(x,y)\mu(dx)\mu(dy),
\]
and $F_1(w)$ is the measurable function on $(0,\infty)^2$ given by
\[
F_1(w)(x,y)=-(x+y)w(x,y)+2\int_0^\infty\int_0^\infty w(x,x')w(y,y')\mu(dx')\mu(dy').
\]
\end{lemma}

\begin{proof}
The initial conditions are satisfied by \autoref{lem:riccati}. We differentiate with respect to $\tau$ and use \autoref{lem:riccati}: 
\begin{align*}
\partial_\tau\psi_0\left(\tau,v^{\otimes2}\right)&=\frac{2}{1-4\psi_0(\tau,v,0)}\partial_\tau\phi_0(\tau,v,0)=F_0\left(\psi_1\left(\tau,v^{\otimes2}\right)\right),
\\
\partial_\tau\psi_1\left(\tau,v^{\otimes2}\right)(x,y)&=\frac{-x\phi_1(\tau,v,0)(x)\phi_1(\tau,v,0)(y)-y\phi_1(\tau,v,0)(x)\phi_1(\tau,v,0)(y)}{1-4\psi_0(\tau,v,0)}\\&\quad+\frac{2\phi_1(\tau,v,0)(x)\phi_1(\tau,v,0)(y)}{\left(1-4\phi_0(\tau,v,0)\right)^2}\left(\int_0^\infty\phi_1(\tau,v,0)(z)\mu(dz)\right)^2\\&=F_1\left(\psi_1\left(\tau,v^{\otimes2}\right)\right)(x,y).
\end{align*}
This concludes the proof.
\end{proof}

\subsection{\texorpdfstring{Affine structure of $(X,\Pi)$}{Affine structure of (X,Pi)}}

The following theorem shows that $(X,\Pi)$ is an affine process with values in $\mathbb R\times L^1(\mu)\otimes_s L^1(\mu)$. The proof is based on an approximation of $\langle Y,u\rangle_\mu$ going back to Carmona, Coutin, and Montseny \cite{carmona2000approximation}. This approximation also provides a mean for simulating the fractional Stein and Stein model.

\begin{theorem}[Affine structure]\label{thm:stein:affine}
Let $\mu$ satisfy \autoref{ass:integrability1} and $(X_0,\Pi_0)\in \mathbb R\times L^1(\mu)\otimes_s L^1(\mu)$. Then $(X,\Pi)$ is an affine process in the sense that for each $0\leq t\leq T$, $u \in i\mathbb R$, and $v^{\otimes2}\in i L^\infty(\mu)\otimes_s L^\infty(\mu)$, the logarithmic conditional characteristic function
\[
\log \mathbb E\left[e^{X_T u + \langle \Pi_T, v^{\otimes2}\rangle_{\mu^{\otimes2}}}\middle|\mathcal F_t\right],
\]
is affine in $(X_t,\Pi_t)$.
\end{theorem}

\begin{proof}
We approximate the measure $\mu$ by a sequence $\mu^n$ of atomic measures. If $\mu^n$ are suitably chosen, it follows from  \cite{carmona2000approximation} that $\langle Y,v\rangle_{\mu^n}$ converges uniformly on compacts in probability (ucp) to $\langle Y,v\rangle_{\mu}$. It follows that $\langle \Pi,v^{\otimes2}\rangle_{(\mu^n)^{\otimes2}} = \langle Y,v\rangle_{\mu^n}^2$ converges ucp to $\langle \Pi,v^{\otimes2}\rangle_{\mu^{\otimes2}} = \langle Y,v\rangle_{\mu}^2$. Let $X^n$ be the corresponding process solving \autoref{equ:stein_SDE} with $\mu$ replaced by $\mu^n$. As stochastic integrals are continuous in the ucp topology, it follows that $X^n_T$ converges in probability to $X_T$. This implies convergence of the logarithmic characteristic function in \autoref{thm:stein:affine}. For each $n$, the logarithm characteristic function is affine by \autoref{rem:sde:pi} and the affine nature of \autoref{equ:stein_SDE}. The result follows.
\end{proof}

\subsection{The uncorrelated case}

By ``uncorrelated'' we mean $d\langle W,\widetilde W\rangle_t=\rho dt = 0$. In the uncorrelated case, the distribution of $X_T$ depends immediately on the integrated variance, which is defined as
\[
\mathrm{IV}(t,T)
=\frac{1}{T-t}\int_t^T \langle Y_s,\lambda\rangle_\mu^2 ds
=\frac{1}{T-t}\int_t^T \langle \Pi_s,\lambda^{\otimes2}\rangle_{\mu^{\otimes2}} ds.
\]  
This dependence is made precise in the following lemma.
	
\begin{lemma}[Conditional CDF]
In the uncorrelated case $\rho=0$, the $\mathcal F_t$-conditional cumulative distribution function of $X_T$ is
\[
\mathbb Q\left[X_T\leq x\middle|\mathcal F_t\right]=\frac{1}{\sqrt{2\pi}}\int_{-\infty}^x\mathbb E\left[\exp\left(-\frac{\left(y-X_t+\frac{T-t}{2}\mathrm{IV}(t,T)\right)^2}{2(T-t)\mathrm{IV}(t,T)}\right)\middle|\mathcal F_t\right]dy,
\]
and the $\mathcal F_t$-conditional characteristic function is
\begin{align*}
\mathbb E\left[e^{X_T u+\langle \Pi_T,v^{\otimes2}\rangle_{\mu^{\otimes2}}}\middle|\mathcal F_t\right]
=
e^{X_tu+\langle \Pi_t,v^{\otimes2}\rangle_{\mu^{\otimes2}}}\mathbb E\left[e^{\frac{T-t}2(u^2-u)\mathrm{IV}(t,T)}\middle|\mathcal F_t\right],
\end{align*}
where $0\leq t\leq T$, $u \in i\mathbb R, v^{\otimes2}\in i L^\infty(\mu)\otimes_s L^\infty(\mu)$.
\end{lemma}

\begin{proof}
This can be seen as in \cite{stein1991} by conditioning on the sigma algebra generated by $(\langle Y_t,v\rangle)_{0\leq t\leq T}$ and by using the independence of $W$ and $\widetilde W$.
\end{proof}

The Fourier transform of the integrated variance process can be calculated explicitly using the affine structure of the process $\Pi$. Thus, in theory, it is possible to characterize the conditional distribution of the integrated variance. An example is given in the next corollary. 

\begin{corollary}[Conditional moments]
For each $0\leq t\leq T$, the first and second $\mathcal F_t$-conditional moments of the integrated variance $\mathrm{IV}(t,T)$ are given by
\[
\begin{aligned}
&\mathbb E\left[\mathrm{IV}(t,T)\middle|\mathcal F_t\right]=\int_t^T\left(2\phi_0(s-t,\lambda,0)+\left\langle\Pi_t,\phi_1(s-t,\lambda,0)^{\otimes2}\right\rangle_{\mu^{\otimes2}}\right)ds,\\
&\mathbb E\left[\mathrm{IV}(t,T)^2\middle|\mathcal F_t\right]=4\int_t^T\int_t^T \Big(\phi_0(s_1\vee s_2-s_1\wedge s_2,\lambda,0)\phi_0(s_1\wedge s_2-t,\lambda,0)
\\&\quad+2\phi_0(s_1\vee s_2-s_1\wedge s_2,\lambda,0)\left\langle\Pi_t,\phi_1(s_1\wedge s_2-t,\lambda,0)^{\otimes2}\right\rangle_{\mu^{\otimes2}}+\frac14 \mathbb E\left[\left\langle \Pi_{s_1\wedge s_2},w(s_1,s_2)\right\rangle^2_{\mu^{\otimes2}}\middle|\mathcal F_t\right]\Big)ds_2ds_1,
\end{aligned}
\]
where $w(s_1,s_2)=\lambda\otimes\phi_1(s_1\vee s_2-s_1\wedge s_2,\lambda,0)+\phi_1(s_1\vee s_2-s_1\wedge s_2,\lambda,0)\otimes \lambda$ is symmetric two tensor and the last expectation is given by \autoref{lem:stein:covariance}.
\end{corollary}

\begin{proof}
We obtain the formula for the conditional mean using \autoref{lem:stein:mean}. Note that we are allowed to exchange the conditional expectation and integration because the integrand is positive. 
For the second moment we use the tower property of conditional expectations and \autoref{lem:stein:mean} for the conditional mean:
\[
\begin{aligned}
&\mathbb E\left[\mathrm{IV}(t,T)^2\middle|\mathcal F_t\right]
=\int_t^T\int_t^T\mathbb E\left[\left\langle\Pi_{s_1},\lambda^{\otimes2}\right\rangle_{\mu^{\otimes2}}\left\langle\Pi_{s_2},\lambda^{\otimes2}\right\rangle_{\mu^{\otimes2}}\middle|\mathcal F_t\right]ds_2ds_1
\\&=\int_t^T\int_t^T\mathbb E\left[\left\langle\Pi_{s_1\wedge s_2},\lambda^{\otimes2}\right\rangle_{\mu^{\otimes2}}\mathbb E\left[\left\langle\Pi_{s_1\vee s_2},\lambda^{\otimes2}\right\rangle_{\mu^{\otimes2}}\middle|\mathcal F_{s_1\wedge s_2}\right]\middle|\mathcal F_t\right]ds_2ds_1
\\&=4\int_t^T\int_t^T\phi_0(s_1\vee s_2-s_1\wedge s_2,\lambda,0)\phi_0(s_1\wedge s_2-t,\lambda,0)
\\&\quad+2\phi_0(s_1\vee s_2-s_1\wedge s_2,\lambda,0)\left\langle\Pi_t,\phi_1(s_1\wedge s_2-t,\lambda,0)^{\otimes2}\right\rangle_{\mu^{\otimes2}}
\\&\quad+\mathbb E\left[\left\langle\Pi_{s_1\wedge s_2},\lambda^{\otimes2}\right\rangle_{\mu^{\otimes2}}\left\langle\Pi_{s_1\wedge s_2},\phi_1(s_1\vee s_2-s_1\wedge s_2,\lambda,0)^{\otimes2}\right\rangle_{\mu^{\otimes2}}\middle|\mathcal F_t\right]ds_2ds_1.
\end{aligned}
\]
We set $s=s_1\wedge s_2$ and $\tau=s_1\vee s_2-s_1\wedge s_2$. Observe that 
\[
\begin{aligned}
&\left\langle\Pi_{s},\lambda^{\otimes2}\right\rangle_{\mu^{\otimes2}}\left\langle\Pi_{s},\phi_1(\tau,\lambda,0)^{\otimes2}\right\rangle_{\mu^{\otimes2}}
=\left\langle Y^{\otimes2}_s,\lambda^{\otimes2}\right\rangle_{\mu^{\otimes2}}\left\langle Y^{\otimes2}_s,\phi_1(\tau,\lambda,0)^{\otimes2}\right\rangle_{\mu^{\otimes2}}
\\&\qquad=\left(\left\langle Y_s,\lambda\right\rangle_{\mu}\left\langle Y_s,\phi_1(\tau,\lambda,0)\right\rangle_{\mu}\right)^2
=\left\langle \Pi_s,\lambda\otimes\phi_1(\tau,\lambda,0)\right\rangle^2_{\mu^{\otimes2}}
=\frac{1}{4}\left\langle \Pi_s,w\right\rangle^2_{\mu^{\otimes2}},
\end{aligned}
\]
where $w=\lambda\otimes\phi_1(\tau,\lambda,0)+\phi_1(\tau,\lambda,0)\otimes \lambda$ is a symmetric two tensor. The result follows from \autoref{lem:stein:covariance}. 
\end{proof}

\section{Proofs and auxiliary results}\label{sec:auxiliary}

\subsection{Stochastic Fubini's theorem}
We refer to the version of the theorem proved in \cite{veraar2012stochastic}. Let $\mu$ be a $\sigma$-finite measure on $(0,\infty)$. Fix $T\geq0$ and denote by $\mathcal P\hspace{-0.1em}r$ the $\sigma$-algebra on $[0,T]\times\Omega$ generated by all progressively measurable processes. 
\begin{theorem}[Stochastic Fubini Theorem]\label{thm:fubini}
Let $G:(0,\infty)\times[0,T]\times\Omega\rightarrow\mathbb R$ be measurable with respect to the product $\sigma$-algebra $\mathcal B(0,\infty)\otimes\mathcal P\hspace{-0.1em}r$. Define processes $\zeta_{1,2}\colon(0,\infty)\times[0,T]\times\Omega\rightarrow\mathbb R$ and $\eta\colon[0,T]\times\Omega\rightarrow\mathbb R$ by
\begin{align*}
\zeta_1(x,t,\omega)&=\int_0^tG(x,s,\omega)ds,
&
\zeta_2(x,t,\omega)&=\left(\int_0^tG(x,s,\cdot)dW_s\right)(\omega),
&
\eta(t,\omega)&=\int_0^\infty G(x,t,\omega)\mu(dx).
\end{align*}
\begin{enumerate}[(i)]
\item Assume $G$ satisfies for almost all $\omega\in\Omega$
\begin{equation}\label{equ:fubini1}
\int_0^\infty\int_0^T\left|G(x,s,\omega)\right|ds\mu(dx)<\infty.
\end{equation}
Then, for almost all $\omega\in\Omega$ and for all $t\in[0,T]$ we have $\zeta_1(\cdot,t,\omega)\in L^1(\mu)$ and
\[
\int_0^{\infty}\zeta_1(x,t,\omega)\mu(dx)=\int_0^t\eta(s,\omega)ds.
\]
\item Assume $G$ satisfies for almost all $\omega\in\Omega$ 
\begin{equation}\label{equ:fubini2}
\int_0^\infty\sqrt{\int_0^TG(x,s,\omega)^2ds}\mu(dx)<\infty.
\end{equation}
Then, for almost all $\omega\in\Omega$ and for all $t\in[0,T]$ we have $\zeta_2(\cdot,t,\omega)\in L^1(\mu)$ and
\[
\int_0^{\infty}\zeta_2(x,t,\omega)\mu(dx)=\left(\int_0^t\eta(s,\cdot)dW_s\right)(\omega).
\] 
\end{enumerate}
\end{theorem}

\begin{remark}
Note that
\[
\int_0^\infty\int_0^T\mathbb E\left[\left|G(x,s)\right|\right]ds\mu(dx)<\infty\quad\text{and}\quad
\int_0^\infty\mathbb E\left[\sqrt{\int_0^TG(x,s)^2ds}\right]\mu(dx)<\infty
\]
imply that conditions \eqref{equ:fubini1} and \eqref{equ:fubini2} hold with probability one.
\end{remark}

\subsection{Reproducing kernel Hilbert spaces}\label{sec:reproducing}

We adapt the exposition of \cite[Section~8]{vanneerven2010gamma} to our setting and refer to this reference for further details. Let $P\colon L^\infty(\mu;\mathbb C)\to L^1(\mu;\mathbb C)$ be a positive and symmetric bounded linear operator, i.e., $\langle Pu,u\rangle_\mu\geq 0$ and $\langle Pu,v\rangle_\mu=\langle Pv,u\rangle_\mu$ for all $u,v\in L^\infty(\mu;\mathbb C)$. The bilinear form $(Pu,Pv)\mapsto \langle Pu,v\rangle_\mu$ defines an inner product on the image of $P$. The completion of the image of $P$ with respect to this inner product is a Hilbert space, which we denote by $\overline{\operatorname{im}(P)}$. The inclusion of the image of $P$ in $L^1(\mu;\mathbb C)$ extends to a bounded injective operator $i\colon \overline{\operatorname{im}(P)}\to L^1(\mu;\mathbb C)$. The space $\mathsf H = \operatorname{im}(i) \subseteq L^1(\mu;\mathbb C)$ with the Hilbert structure induced by the bijection $i\colon \overline{\operatorname{im}(P)}\to \mathsf H$ is called the reproducing kernel Hilbert space\footnote{In \cite{vanneerven2010gamma} the space $\overline{\operatorname{im}(P)}$ is called reproducing kernel Hilbert space of $P$.} of $P$. If $u,v \in L^\infty(\mu;\mathbb C)$, then $Pu,Pv \in \mathsf H$ and $\langle Pu,Pv\rangle_{\mathsf H} = \langle Pu,v\rangle_\mu$, where the inclusion $i$ is dropped from our notation.

\subsection{Gaussian measures on Banach spaces}\label{sec:gamma}

We present those parts of the theory of $\gamma$-radonifying operators which are used in the proof of \autoref{thm:banach}. For the purpose of this section, let $E$ be a separable Banach space, let $H$ be a separable Hilbert space, which we identify with its dual, and let $T \in (0,\infty)$. 

\begin{definition}[{\cite[Definition~3.7]{vanneerven2010gamma}}]\label{def:gamma:gamma}
The Banach space $\gamma(H;E)$ of $\gamma$-radonifying operators from $H$ to $E$ is defined as the completion of the algebraic tensor product $H\otimes E$ with respect to the norm
\begin{equation*}
\left\|\sum_{n=1}^N h_n \otimes x_n \right\|_{\gamma(H,E)}^2 
=
\mathbb E\left[\left\|\sum_{n=1}^n \gamma_n x_n\right\|^2\right],
\end{equation*}
where it is assumed that $h_1,\dots,h_n$ are orthonormal in $H$ and $\gamma_1,\dots,\gamma_n$ are i.i.d.\@ standard normal. 
\end{definition}

\begin{theorem}[{\cite[Theorem~7.4]{vanneerven2010gamma}}]\label{thm:gamma:gamma}
Let $i \in L(H;E)$. Then $i \in \gamma(H;E)$ if and only if there exists a centered Gaussian $E$-valued random variable $X$ satisfying 
\begin{equation*}
\mathbb E\left[\langle X,x^*\rangle_{E,E^*}^2\right] = \|i^*x^*\|_H^2, \qquad x^* \in E^*.
\end{equation*}
In this situation we have $\|i\|_{\gamma(H,E)}=\mathbb E[\|X\|_E^2]$.
\end{theorem}

The Bochner space of strongly measurable functions $\Theta\colon(0,T]\to H$ satisfying $\int_{(0,T]}\|\Theta(s)\|_H^2ds<\infty$ is denoted by $L^2((0,T];H)$.

\begin{theorem}\label{thm:gamma:ou}
Let $\Theta\colon(0,T]\to L(H;E)$ be a function such that for all $x^*\in E^*$ the function $t \mapsto \Theta(t)^*x^*$ belongs to $L^2((0,T];H)$. Then the following statements hold:
\begin{enumerate}[(i)]
\item \label{item:gamma:ou1}
For each $t \in (0,T]$ there is a unique positive symmetric linear operator $P_t \in L(E^*;E)$ such that for all $x^*,y^* \in E^*$, 
\begin{equation*}
\langle P_tx^*,y^*\rangle_{E,E^*} = \int_0^t \langle\Theta^*(s)x^*,\Theta^*(s) y^*\rangle_H ds.
\end{equation*}

\item \label{item:gamma:ou2}
Let $\mathsf H_T \subseteq E$ be the reproducing kernel Hilbert space of $P_T$. Then the inclusion of $\mathsf H_T$ into $E$ is $\gamma$-radonifying if and only if there exists a predictable process $X\colon\Omega\times[0,T]\to E$ which satisfies for all $s,t \in [0,T]$ and $x^*,y^* \in E^*$ that
\begin{equation*}
\mathbb E\left[\langle X_t,x^*\rangle_{E,E^*} \langle X_s,y^*\rangle_{E,E^*}\right]
=
\int_0^{t\wedge s} \langle \Theta(t-u)^*x^*,\Theta(s-u)^*x^*\rangle_H du.
\end{equation*}
In this situation, $X$ is called an OU process associated to $\Theta$.
\end{enumerate}
\end{theorem}

\begin{proof}
(\ref{item:gamma:ou1}) is shown in \cite[Lemma~2.1 and Proposition~2.2]{brzezniak2000stochastic}. The necessary part of (\ref{item:gamma:ou2}) is shown in \cite[Proposition~2.8]{brzezniak2000stochastic} and the sufficient part in \cite[Theorem~3.3]{brzezniak2000stochastic}.
\end{proof}

\subsection{Basic estimates}

We collect some inequalities and estimates which are used throughout the paper. 

\begin{lemma}[Elementary inequalities]\label{lem:elem_ineq}
The following inequalities hold true for all $x,y>0$ 
\begin{align}
\label{equ:elem_ineq1}
1\wedge xy&\leq\left(1\vee x\right)\left(1\wedge y\right),\\
\label{equ:elem_ineq2}
y\wedge x^{-1}&\leq\left(1\vee y\right)\left(1\wedge x^{-1}\right),
\end{align}
and for all $\alpha,\tau>0$ and $0<\epsilon<1$,
\begin{align}
\label{equ:elem_ineq3}
&e^{-x\tau}\leq\left(1\vee\left(\frac{\tau}{\alpha}\right)^{-\alpha}\right)\left(1\wedge x^{-\alpha}\right),\\
\label{equ:elem_ineq4}
&\frac{1-e^{-\tau x}}{x}\leq\left(1\vee\tau\right)\left(1\wedge x^{-1}\right),\\
\label{equ:elem_ineq5}
&\frac{1-e^{-\tau x}(1+\tau x)}{x^2}\leq\left(1\vee\tau^2\right)\left(1\wedge x^{-2}\right),\\
\label{equ:elem_ineq6}
&\frac{1-e^{-\tau x}\left(1+\tau x+\frac{1}{2}\tau^2x^2\right)}{x^3}\leq\left(1\vee\tau^3\right)\left(1\wedge x^{-3}\right),
\\
\label{equ:elem_ineq7}
&\int_0^t s^{-\epsilon} e^{-2xs} ds \leq \left(2^{\epsilon-1}\Gamma(1-\epsilon)\vee \frac{t^{1-\epsilon}}{1-\epsilon}\right)\left(1\wedge x^{\epsilon-1}\right).
\end{align}
\end{lemma}

\begin{proof}
For the inequalities \eqref{equ:elem_ineq1}-\eqref{equ:elem_ineq2} consider the following four cases separately.
\begin{enumerate}
\item If $0<x,y\leq1$. Then, $1\wedge xy=xy\leq y=\left(1\vee x\right)\left(1\wedge y\right)$ and $y\wedge x^{-1}=y\leq1=\left(1\vee y\right)\left(1\wedge x^{-1}\right)$.
\item If $0<x\leq1\leq y$. Then, $1\wedge xy\leq1=\left(1\vee x\right)\left(1\wedge y\right)$ and $y\wedge x^{-1}\leq y=\left(1\vee y\right)\left(1\wedge x^{-1}\right)$.
\item If $0<y\leq1\leq x$. Then, $1\wedge xy\leq xy=\left(1\vee x\right)\left(1\wedge y\right)$ and $y\wedge x^{-1}\leq x^{-1}=\left(1\vee y\right)\left(1\wedge x^{-1}\right)$.
\item If $1\leq x,y$. Then, $1\wedge xy=1\leq x=\left(1\vee x\right)\left(1\wedge y\right)$ and $y\wedge x^{-1}=x^{-1}\leq yx^{-1}=\left(1\vee y\right)\left(1\wedge x^{-1}\right)$.
\end{enumerate}
Consider the functions $f(x,\tau)=e^{-x\tau}$ and $g(x,\tau,\alpha)=x^\alpha f(x,\tau)$. Obviously, $f(x,\tau)\leq1$ for all $x,\tau>0$. Note that $\partial_xg(x,\tau,\alpha)=x^{\alpha-1}e^{-x\tau}\left(\alpha-\tau x\right)$ and $g$ attains its maximum in $x$ at $\frac{\alpha}{\tau}$. Hence, \autoref{equ:elem_ineq3} follows from
\[
f(x,\tau)=\frac{g(x,\tau,\alpha)}{x^\alpha}\leq\frac{g\left(\frac{\alpha}{\tau},\tau,\alpha\right)}{x^\alpha}=\left(\frac{\tau}{\alpha}x\right)^{-\alpha}e^{-\alpha}\leq\left(\frac{\tau}{\alpha}x\right)^{-\alpha},
\]
and \autoref{equ:elem_ineq1}.

Define $k_1(x,\tau)=\frac{1-e^{-\tau x}}{x}$, $k_2(x,\tau)=\frac{1-e^{-\tau x}(1+\tau x)}{x^2}$ and $k_3(x,\tau)=\frac{1-e^{-\tau x}\left(1+\tau x+\frac{1}{2}\tau^2x^2\right)}{x^3}$. Computing the derivatives with respect to $x$ shows that $k_{1,2,3}(\cdot,\tau)$ are decreasing functions in $x$ for all $\tau>0$. The inequalities \eqref{equ:elem_ineq4}-\eqref{equ:elem_ineq6} follow from
\[
\lim_{x\rightarrow\infty}k_{1,2,3}(x,\tau)=0,\quad\lim_{x\rightarrow0^+}k_i(x,\tau)=\begin{cases}\tau,\: &i=1,\\\frac{\tau^2}{2},\: &i=2,\\ \frac{\tau^3}{6},\: &i=3,\end{cases}
\]
and \autoref{equ:elem_ineq2}. \autoref{equ:elem_ineq7} follows from the relation
\begin{align*}
\int_0^t s^{-\epsilon}e^{-2sx} ds
&=
\int_0^{2tx}(2x)^{\epsilon-1} s^{-\epsilon}e^{-s}ds,
\end{align*}
and from the following two estimates:
\begin{align*}
\int_0^{2tx}(2x)^{\epsilon-1} s^{-\epsilon}e^{-s}ds
&\leq \int_0^\infty(2x)^{\epsilon-1} s^{-\epsilon}e^{-s}ds = (2x)^{\epsilon-1}\Gamma(1-\epsilon),
\\
\int_0^{2tx}(2x)^{\epsilon-1} s^{-\epsilon}e^{-s}ds
&\leq 
\int_0^{2tx}(2x)^{\epsilon-1} s^{-\epsilon}ds
=
\frac{t^{1-\epsilon}}{1-\epsilon}.
\end{align*}
This concludes the proof.
\end{proof}

\begin{lemma}[Integrability of elementary expressions]\label{lem:elem_int}
Let \autoref{ass:integrability1} be in place and let $\tau,\alpha>0$. Then
\begingroup\allowdisplaybreaks\begin{align}
\label{equ:elem_int1}
&\int_0^\infty e^{-x\tau}\mu(dx)<\infty,\\
\label{equ:elem_int2}
&\int_0^\infty e^{-x\tau}\nu(dx)<\infty,\\
\label{equ:elem_int3}
&\int_0^\infty x^{\alpha}e^{-x\tau}\mu(dx)<\infty,\\
\label{equ:elem_int4}
&\int_0^\infty x^{\alpha}e^{-x\tau}\nu(dx)<\infty,\\
\label{equ:elem_int5}
&\int_0^\infty\sqrt{\frac{1-e^{-2\tau x}}{x}}\mu(dx)<\infty,\\
\label{equ:elem_int6}
&\int_0^\infty\sqrt{\frac{1-e^{-2\tau x}\left(1+2\tau x+2\tau^2x^2\right)}{x^3}}\nu(dx)<\infty,\\
\label{equ:elem_int7}
&\int_0^\infty\sqrt{\frac{1-2e^{-\tau x}(\tau x+1)(1-\tau xe^{-\tau x})+e^{-2\tau x}}{x^3}}\mu(dx)<\infty.
\end{align}\endgroup
Furthermore, for each $0\leq t< T$ we have
\begingroup\allowdisplaybreaks\begin{align}
\label{equ:elem_int8}
&\int_0^\infty\sqrt{\int_t^T e^{-2x(T-s)}ds}\mu(dx)<\infty,\\
\label{equ:elem_int9}
&\int_0^\infty\sqrt{\int_t^T(T-s)^2 e^{-2x(T-s)}ds}\nu(dx)<\infty,\\
\label{equ:elem_int10}
&\int_0^\infty\int_t^Te^{-x(T-s)}ds\mu(dx)<\infty,\\
\label{equ:elem_int11}
&\int_0^\infty\int_t^T(T-s)e^{-x(T-s)}ds\nu(dx)<\infty,\\
\label{equ:elem_int12}
&\int_0^\infty\int_t^T\frac{1-e^{-x(T-s)}}{x}(1\wedge x^{-\frac12})ds\nu(dx)<\infty,\\
\label{equ:elem_int13}
&\int_0^\infty\sqrt{\int_t^T\left(\frac{1-e^{-x(T-s)}}{x}\right)^2ds}\mu(dx)<\infty,\\
\label{equ:elem_int14}
&\int_0^\infty\sqrt{\int_t^T\left(1-\frac{e^{-x(T-s)}\left(1+x(T-s)\right)}{x^2}\right)^2ds}\mu(dx)<\infty,\\
\label{equ:elem_int15}
&\int_0^\infty\int_0^\infty\int_t^Te^{-(x+y)(T-s)}ds\mu(dx)\mu(dy)<\infty,\\
\label{equ:elem_int16}
&\int_0^\infty\int_0^\infty\int_t^T(T-s)^2e^{-(x+y)(T-s)}ds\nu(dx)\nu(dy)<\infty.
\end{align}\endgroup
\end{lemma}

\begin{proof}
Equations \eqref{equ:elem_int1} and \eqref{equ:elem_int2} follow directly from \eqref{equ:elem_ineq3} for $\alpha=\frac{1}{2}$ and $\alpha=\frac{3}{2}$, respectively. Applying \autoref{equ:elem_ineq3} for $\beta>\alpha$ we obtain
\[
\begin{aligned}
\int_0^\infty x^\alpha e^{-x\tau}\mu(dx)&\leq\int_0^1e^{-x\tau}\mu(dx)+\int_1^\infty x^{\alpha}e^{-x\tau}\mu(dx)\\&\leq\int_0^1\left(1\vee\left(\frac{\tau}{\beta}\right)^{-\beta}\right)\mu(dx)+\int_1^\infty x^{\alpha-\beta}\left(1\vee\left(\frac{\tau}{\beta}\right)^{-\beta}\right)\mu(dx)\\&=\left(1\vee\left(\frac{\tau}{\beta}\right)^{-\beta}\right)\int_0^\infty\left(1\wedge x^{\alpha-\beta}\right)\mu(dx),
\end{aligned}
\]
and in the same way $\int_0^\infty x^\alpha e^{-x\tau}\nu(dx)\leq(1\vee(\frac{\tau}{\beta})^{-\beta})\int_0^\infty\left(1\wedge x^{\alpha-\beta}\right)\nu(dx)$. Setting $\beta=\alpha+\frac{1}{2}$ and $\beta=\alpha+\frac{3}{2}$ one proves \eqref{equ:elem_int3} and \eqref{equ:elem_int4}, respectively.
By \autoref{equ:elem_ineq4} we obtain \autoref{equ:elem_int5}
\begin{equation}\label{equ:elem_int5_detail}
\begin{aligned}
&\int_0^\infty\sqrt{\frac{1-e^{-2\tau x}}{x}}\mu(dx)
\leq\left(1\vee \left(2\tau\right)^{\frac{1}{2}}\right)\int_0^\infty\left(1\wedge x^{-\frac{1}{2}}\right)\mu(dx)<\infty.
\end{aligned}
\end{equation}
By \autoref{equ:elem_int5} we obtain \autoref{equ:elem_int8}
\[
\int_0^\infty \sqrt{\int_t^Te^{-2x(T-s)}ds}\mu(dx)=\int_0^\infty \sqrt{\frac{1-e^{-2(T-t)x}}{2x}}\mu(dx)<\infty.
\]
By \autoref{equ:elem_ineq6} we obtain \autoref{equ:elem_int6}
\begin{equation}\label{equ:elem_int6_detail}
\begin{aligned}
&\int_0^\infty \sqrt{\frac{1-e^{-2\tau x}\left(1+2\tau x+2\tau^2x^2\right)}{x^3}}\nu(dx)
\leq(1\vee (2\tau)^{\frac{3}{2}})\int_0^\infty (1\wedge x^{-\frac{3}{2}})\nu(dx)<\infty.
\end{aligned}
\end{equation}
\autoref{equ:elem_int7} follows from
\[
\begin{aligned}
&\int_0^\infty \sqrt{\frac{1-2 e^{-\tau x} (\tau x+1)+2\tau x e^{-2\tau x} (\tau x+1)+e^{-2\tau x}}{4x^3}}\nu(dx)
\\&\quad
\leq\int_0^{1/\tau} \sqrt{\frac{\tau^2}{6 x}}\nu(dx)+\int_{1/\tau}^\infty \sqrt{\frac{2}{x^3}}\nu(dx)
\leq\sqrt{2}(\tau\vee 1)\int_0^\infty (x^{-\frac12}\wedge x^{-\frac32})\nu(dx)<\infty.
\end{aligned}
\]
\autoref{equ:elem_int6} implies \autoref{equ:elem_int9}
\[
\begin{aligned}
&\int_0^\infty \sqrt{\int_t^T (T-s)^2 e^{-2x(T-s)}ds}\nu(dx)
=
\int_0^\infty \sqrt{\frac{1-e^{-2(T-t)x}\left(1+2(T-t)x+2(T-t)^2x^2\right)}{4 x^3}}\nu(dx)<\infty.
\end{aligned}
\]
\autoref{equ:elem_int10} is obtained using \eqref{equ:elem_ineq3} for $\alpha=\frac{1}{2}$
\[
\begin{aligned}
\int_0^\infty\int_t^Te^{-x(T-s)}ds\mu(dx)
&\leq\int_t^T(1\vee(T-s)^{-\frac{1}{2}})ds\int_0^\infty(1\wedge x^{-\frac{1}{2}})\mu(dx)
\\&=
\left(t\vee(T-1)-t+2\sqrt{T-\left(t\vee(T-1)\right)}\right)\int_0^\infty(1\wedge x^{-\frac{1}{2}})\mu(dx)<\infty.
\end{aligned}
\]
\autoref{equ:elem_int11} is obtained using \eqref{equ:elem_ineq3} for $\alpha=\frac{3}{2}$
\[
\begin{aligned}
\int_0^\infty\int_t^T(T-s)e^{-x(T-s)}ds\nu(dx)
&\leq
\int_t^T(T-s)\left(1\vee(T-s)^{-\frac{3}{2}}\right)ds\int_0^\infty(1\wedge x^{-\frac{3}{2}})\mu(dx)
\\&\leq
\left(\int_t^T(T-s)ds \vee \int_t^T (T-s)^{-\frac12}ds\right)\int_0^\infty(1\wedge x^{-\frac{3}{2}})\mu(dx)
\\&=
\left(\frac{(T-t)^2}{2}\vee 2\sqrt{T-t}\right)\int_0^\infty(1\wedge x^{-\frac{3}{2}})\mu(dx)<\infty.
\end{aligned}
\]
\autoref{equ:elem_ineq4} immediately implies \autoref{equ:elem_int12}
\[
\begin{aligned}
&\int_0^\infty\int_t^T\frac{1-e^{-x(T-s)}}{x}(1\wedge x^{-\frac12})ds\nu(dx)
\leq\int_0^\infty\left(1\wedge x^{-\frac{3}{2}}\right)\nu(dx)\int_t^T\left(1\vee(T-s)\right)ds<\infty,
\end{aligned}
\]
and \autoref{equ:elem_int13}
\[
\begin{aligned}
&\int_0^\infty\sqrt{\int_t^T\left(\frac{1-e^{-x(T-s)}}{x}\right)^2ds}\mu(dx)
\leq\sqrt{\int_t^T\left(1\vee(T-s)^2\right)ds}\int_0^\infty\left(1\wedge\frac{1}{x}\right)\mu(dx)<\infty.
\end{aligned}
\]
\autoref{equ:elem_ineq5} immediately implies \autoref{equ:elem_int14}
\[
\begin{aligned}
&\int_0^\infty\sqrt{\int_t^T\left(\frac{e^{-x(T-s)}(1+x(T-s))-1}{x^2}\right)^2ds}\mu(dx)
\leq\sqrt{\int_t^T\left(1\vee(T-s)^4\right)ds}\int_0^\infty\left(1\wedge x^{-2}\right)\mu(dx)<\infty.
\end{aligned}
\]
\autoref{equ:elem_int15} follows from \autoref{equ:elem_int8} applying Cauchy-Schwarz inequality
\[
\begin{aligned}
\int_0^\infty\int_0^\infty\int_t^Te^{-(x+y)(T-s)}ds\mu(dx)\mu(dy)
&\leq\int_0^\infty\int_0^\infty\sqrt{\int_t^Te^{-2x(T-s)}ds}\sqrt{\int_t^Te^{-2y(T-s)}ds}\mu(dx)\mu(dy)
\\&=\left(\int_0^\infty\sqrt{\int_t^Te^{-2x(T-s)}ds}\mu(dx)\right)^2<\infty.
\end{aligned}
\]
In the same way \autoref{equ:elem_int16} follows from \autoref{equ:elem_int9}
\begin{equation*}
\int_0^\infty\int_0^\infty\int_t^T(T-s)^2e^{-(x+y)(T-s)}ds\nu(dx)\nu(dy)
\leq\left(\int_0^\infty\sqrt{\int_t^T(T-s)^2e^{-2y(T-s)}ds}\right)^2<\infty.\qedhere
\end{equation*}
\end{proof}

\subsection{Auxiliary results for \autoref{sec:ou}}

\begin{lemma}[Conditional moments of $(Y,Z)$]\label{lem:representation_ou}
For each $x\in(0,\infty)$ and $0\leq t\leq T$, the process $(Y^x,Z^x)$ can be represented as
\begin{equation}\label{equ:representation_ou2}
\begin{aligned}
Y^x_T&=Y^x_te^{-(T-t)x}+\int_t^Te^{-(T-s)x}dW_s,\\
Z^x_T&=Z^x_te^{-(T-t)x}+Y^x_t(T-t)e^{-(T-t)x}+\int_t^T(T-s)e^{-(T-s)x}dW_s.
\end{aligned}
\end{equation}
The random variables $Y^{x}_T$ and $Z^{x}_T$ have conditional means given by
\[
\mathbb E\left[Y^{x}_T\middle|\mathcal F_t\right]=Y^x_te^{-(T-t)x},\quad\mathbb E\left[Z^{x}_T\middle|\mathcal F_t\right]=Z^x_te^{-(T-t)x}+Y^x_t(T-t)e^{-(T-t)x}.
\]
Moreover, for $x_1,x_2\in(0,\infty)$ we have conditional covariances
\[
\begin{aligned}
\mathrm{Cov}\left(Y^{x_1}_T,Y^{x_2}_T\middle|\mathcal F_t\right)&=\frac{1-e^{-(T-t)(x_1+x_2)}}{x_1+x_2},\\\mathrm{Cov}\left(Y^{x_1}_T,Z^{x_2}_T\middle|\mathcal F_t\right)&=\frac{1-e^{-(T-t)(x_1+x_2)}\left(1+(T-t)\left(x_1+x_2\right)\right)}{\left(x_1+x_2\right)^2},\\\mathrm{Cov}\left(Z^{x_1}_T,Z^{x_2}_T\middle|\mathcal F_t\right)&=\frac{2-e^{-(T-t)(x_1+x_2)}\left(2+2(T-t)(x_1+x_2)+(T-t)^2(x_1+x_2)^2\right)}{\left(x_1+x_2\right)^3}.
\end{aligned}
\]
\end{lemma}

\begin{proof}
The representation in \autoref{equ:representation_ou2} can be deduced from the SDE \eqref{equ:ou_sde} for $(Y^x,Z^x)$ using \autoref{thm:fubini}(ii)
\[
\begin{aligned}
Z^x_T&=Z^x_te^{-(T-t)x}+\int_t^T e^{-(T-s)x}\left(Y^x_te^{-(s-t)x}+\int_t^se^{-(s-u)x}dW_u\right)ds\\&=Z^x_te^{-(T-t)x}+Y^x_t(T-t)e^{-(T-t)x}+\int_t^T\int_t^se^{-(T-u)x}dW_uds\\&=Z^x_te^{-(T-t)x}+Y^x_t(T-t)e^{-(T-t)x}+\int_t^T\int_u^Te^{-(t-u)x}dsdW_u\\&=Z^x_te^{-(T-t)x}+Y^x_t(T-t)e^{-(T-t)x}+\int_t^T(T-u)e^{-(T-u)x}dW_u.
\end{aligned}
\]
The condition \eqref{equ:fubini2} is satisfied because $\int_t^T\sqrt{\int_t^s e^{-2(t-u)x}du}ds<\infty$. The conditional means can be read off directly from the representation of $(Y^x,Z^x)$. The formulas for the conditional covariances are obtained using It\=o's isometry by calculating the following integrals
\begin{align*}
\operatorname{Cov}\left(Y^{x_1}_T,Y^{x_2}_T\middle|\mathcal F_t\right)&=\int_t^{T}e^{-(T-s)(x_1+x_2)}ds,\\
\operatorname{Cov}\left(Y^{x_1}_T,Z^{x_2}_T\middle|\mathcal F_t\right)&=\int_t^{T}(T-s)e^{-(T-s)(x_1+x_2)}ds,\\
\operatorname{Cov}\left(Z^{x_1}_T,Z^{x_2}_T\middle|\mathcal F_t\right)&=\int_t^{T}(T-s)^2e^{-(T-s)(x_1+x_2)}ds.	
\qedhere
\end{align*}
\end{proof}

\begin{lemma}[Integrability of $(Y,Z)$]\label{lem:integrability_ou}
Let \autoref{ass:integrability1} be in place and assume $(Y_0,Z_0)\in L^1(\mu)\times L^1(\nu)$ a.s. Then, for each $t\geq0$, $Y_t \in L^1(\mu)$ and $Z_t \in L^1(\nu)$ holds with probability one.
\end{lemma}

\begin{proof}
By \autoref{lem:representation_ou} we have for $(Y^x,Z^x)$
\begin{align*}
Y^x_t&=Y^x_0e^{-tx}+\int_0^te^{-(t-s)x}dW_s,
&
Z^x_t&=Z^x_0e^{-tx}+Y^x_0te^{-tx}+\int_0^t(t-s)e^{-(t-s)x}dW_s.
\end{align*}
The deterministic parts are integrable because 
\[
\begin{aligned}
&\int_0^\infty|Y^x_0|e^{-tx}\mu(dx)\leq\|Y_0\|_{L^1(\mu)}<\infty,\\
&\int_0^\infty|Z^x_0|e^{-tx}\nu(dx)\leq\|Z_0\|_{L^1(\nu)}<\infty,\\
&\int_0^\infty|Y^x_0|te^{-tx}\nu(dx)\leq\sup_{x\in(0,\infty)}\left(p(x)e^{-xt}\right)t\|Y_0\|_{L^1(\mu)}<\infty,
\end{aligned}
\]
where \autoref{ass:density1} is used in the last line. Therefore we can assume without loss of generality that $(Y_0,Z_0)$ vanish. Then for each $t\geq 0$, 
\begingroup\allowdisplaybreaks\begin{align*}
\mathbb E\left[\|Y_t\|_{L^1(\mu)}\right]
&= 
\int_0^\infty \mathbb E\left[|Y^x_t|\right]\mu(dx)
=
\frac{\sqrt2}{\sqrt\pi}\int_0^\infty \sqrt{\operatorname{Var}(Y^x_t)}\mu(dx)
=
\frac{\sqrt2}{\sqrt\pi}\int_0^\infty \sqrt{\frac{1-e^{-2tx}}{2x}}\mu(dx)<\infty,
\\
\mathbb E\left[\|Z_t\|_{L^1(\nu)}\right]
&= 
\int_0^\infty \mathbb E\left[|Z^x_t|\right]\nu(dx)
=
\frac{\sqrt2}{\sqrt\pi}\int_0^\infty \sqrt{\operatorname{Var}(Z^x_t)}\nu(dx)
\\&=
\frac{\sqrt2}{\sqrt\pi}\int_0^\infty \sqrt{\frac{1-e^{-2tx}\left(1+2tx+2t^2x^2\right)}{4 x^3}}\nu(dx)<\infty,
\end{align*}\endgroup
which follows from Equations~\eqref{equ:elem_int5} and \eqref{equ:elem_int7}. Therefore, $Y_t \in L^1(\mu)$ and $Z_t\in L^1(\nu)$ holds almost surely.
\end{proof}

\begin{lemma}[Linear functionals of $(Y,Z)$]\label{lem:representation_inf_dim_ou}
Let \autoref{ass:integrability1} be in place and assume $(Y_0,Z_0)\in L^1(\mu)\times L^1(\nu)$ a.s. Then the process $(Y,Z)$ satisfies for each $0\leq t\leq T$ and $(u,v)\in L^\infty(\mu;\mathbb C)\times L^\infty(\nu;\mathbb C)$
\[
\begin{aligned}
\left\langle Y_T,u\right\rangle_\mu&=\int_0^\infty Y^x_t e^{-(T-t)x}u(x)\mu(dx)+\int_t^T\int_0^\infty e^{-x(T-s)}u(x)\mu(dx)dW_s,\\
\left\langle Z_T,v\right\rangle_\nu&=\int_0^\infty \left(Z^x_t e^{-(T-t)x}+Y^x_t(T-t)e^{-(T-t)x}\right)v(x)\nu(dx)+\int_t^T\int_0^\infty(T-s)e^{-x(T-s)}v(x)\nu(dx)dW_s.
\end{aligned}
\]
In particular, the random variable $\left\langle Y_T,u\right\rangle_\mu+\left\langle Z_T,v\right\rangle_\nu$ is Gaussian, given $\mathcal F_t$.
\end{lemma}

\begin{proof}
The statement follows from \autoref{lem:representation_ou,lem:integrability_ou} and from \autoref{thm:fubini}. Condition \eqref{equ:fubini2} of the stochastic Fubini theorem are satisfied by Equations \eqref{equ:elem_int8} and \eqref{equ:elem_int9}.
\end{proof}

\begin{lemma}[Covariance operators]\label{lem:cov}
Let \autoref{ass:integrability1} be in place and $(Y_0,Z_0)\in L^1(\mu)\times L^1(\nu)$ a.s. Then for all $(u_1,u_2)\in L^\infty(\mu;\mathbb C)\times L^\infty(\mu;\mathbb C)$, $(v_1,v_2)\in L^\infty(\nu;\mathbb C)\times L^\infty(\nu;\mathbb C)$ and for all $0\leq t\leq T$
\begin{align*}
\operatorname{Cov}\left(\left\langle Y_T,u_1\right\rangle_\mu,\left\langle Y_T,u_2\right\rangle_\mu\middle|\mathcal F_t\right)&=\left\langle P_{T-t}u_1,u_2\right\rangle_\mu,
&
\operatorname{Cov}\left(\left\langle Z_T,v_1\right\rangle_\mu,\left\langle Z_T,v_2\right\rangle_\mu\middle|\mathcal F_t\right)&=\left\langle Q_{T-t}v_1,v_2\right\rangle_\nu,
\end{align*}
where $P_\tau\colon L^\infty(\mu;\mathbb C)\to L^1(\mu;\mathbb C)$ and $Q_\tau\colon L^\infty(\nu;\mathbb C)\to L^1(\nu; \mathbb C)$ are bounded linear operators given by  
\[
\begin{aligned}
P_\tau u(x)&=\int_0^\infty \frac{1-e^{-\tau(x+y)}}{x+y} u(y)\mu(dy),\\
Q_\tau v(x)&=\int_0^\infty \frac{2 - e^{-\tau(x + y) } (2 + 
    2 \tau(x + y)  + \tau^2 (x + y)^2 )}{(x + y)^3}v(y)\nu(dy),
\end{aligned}
\]
for $u\in L^\infty(\mu;\mathbb C)$, $v\in L^\infty(\nu;\mathbb C)$ and $\tau\geq0$. In particular, $Y_T$ and $Z_T$ are Gaussian random variables, given $\mathcal F_t$, with covariance operators $P_{T-t}$ and $Q_{T-t}$, respectively.
\end{lemma}

\begin{proof}
For each $t\geq0$ and any $u_{1,2}\in L^\infty(\mu)$ and $v_{1,2}\in L^\infty(\nu)$ we have using the representation of \autoref{lem:representation_inf_dim_ou}
\[
\begin{aligned}
&\operatorname{Cov}\left(\left\langle Y_T,u_1\right\rangle_\mu,\left\langle Y_T,u_2\right\rangle_\mu\middle|\mathcal F_t\right)=\int_0^\infty\int_0^\infty \operatorname{Cov}\left(Y^x_T,Y^y_T\middle|\mathcal F_t\right)u_1(x)u_2(y)\mu(dy)\mu(dx)\\&\quad=\int_0^\infty\int_0^\infty\int_t^Te^{-(T-s)x}e^{-(T-s)y}dsu_1(x)u_2(y)\mu(dy)\mu(dx)
=\left\langle P_{T-t}u_1,u_2\right\rangle_\mu,\\&\operatorname{Cov}\left(\left\langle Z_T,v_1\right\rangle_\mu,\left\langle Z_T,v_2\right\rangle_\mu\middle|\mathcal F_t\right)=\int_0^\infty\int_0^\infty \operatorname{Cov}\left(Z^x_T,Z^y_T\middle|\mathcal F_t\right)v_1(x)v_2(y)\mu(dy)\mu(dx)\\&\quad=\int_0^\infty\int_0^\infty\int_t^T(T-s)^2e^{-(T-s)x}e^{-(T-s)y}dsv_1(x)v_2(y)\nu(dy)\nu(dx)=\left\langle Q_{T-t}v_1,v_2\right\rangle_\nu
\end{aligned}
\]
By Equations \eqref{equ:elem_int15} and \eqref{equ:elem_int16} we have 
\[
\begin{aligned}
\|P_\tau u\|_{L^1(\mu)}&=\int_0^\infty\int_0^\infty\int_0^\tau e^{-s(x+y)}|u(x)|ds\mu(dx)\mu(dy)
\leq C \|u\|_{L^\infty(\mu)}<\infty,\\
\|Q_\tau v\|_{L^1(\nu)}&=\int_0^\infty\int_0^\infty\int_0^\tau s^2e^{-s(x+y)}|v(x)|ds\nu(dx)\nu(dy)
\leq C \|u\|_{L^\infty(\mu)}<\infty,
\end{aligned}
\]
for some constant $C$. The last two inequalities imply that $P_\tau\colon L^\infty(\mu;\mathbb C)\to L^1(\mu;\mathbb C)$ and $Q_\tau\colon L^\infty(\nu;\mathbb C)\to L^1(\nu;\mathbb C)$ are bounded linear operators. 
\end{proof}

\begin{lemma}[Maximum inequality for OU processes]\label{lem:ou_max_ineq}
There exists a constant $C>0$ such that for each $t\geq 0$ and $x>0$
\begin{align*}
\mathbb E\left[\sup_{s\in[0,t]} |Y^x_s| \right]\leq C(\log(1+tx))^{1/2} x^{-1/2},
&&
\mathbb E\left[\sup_{s\in[0,t]}|Z^x_s|\right]\leq C(\log(1+tx))^{1/2}x^{-3/2}.
\end{align*}
holds for the processes $(Y^x,Z^x)$ with initial value $(Y^x_0,Z^x_0)=(0,0)$.
\end{lemma}

\begin{proof}
The inequality for $Y^x$ follows from the maximal inequalities for OU processes developed by Graversen and Peskir \cite{graversen2000maximal}. For the process $Z^x$, we estimate for each $t\geq 0$ and $x>0$
\begin{align*}
\mathbb E\left[\sup_{s\in[0,t]}|Z^x_s|\right]
&\leq
\mathbb E\left[\int_0^te^{-(t-s)x}|Y^x_s|ds\right]
\leq 
C\int_0^t e^{-(t-s)x}(\log(1+sx))^{1/2}x^{-1/2}ds
\\&=
C\left[e^{-(t-s)x} (\log(1+sx))^{1/2}x^{-3/2}\right]_0^t
-\frac C2 \int_0^t e^{-(t-s)x}(\log(1+sx))^{-1/2}(1+sx)^{-1}x^{-1/2}ds
\\&\leq
C\left[e^{-(t-s)x} (\log(1+sx))^{1/2}x^{-3/2}\right]_0^t
= 
C(\log(1+tx))^{1/2}x^{-\frac{3}{2}}. 
\qedhere
\end{align*}
\end{proof}

\begin{lemma}[Auxiliary estimates for semimartingale decomposition]\label{lem:semimartingale_fubini}
Let $G(x,t)$ be deterministic and jointly measurable in $(x,t)\in(0,\infty)\times[0,\infty)$. Assume $Y_0=Z_0=0$. Then, with probability one,
\[
\begin{aligned}
&\int_0^{\infty}\int_0^t|G(x,t)Y^x_s(\omega)|ds\mu(dx)\leq(1\vee t^{\frac{1}{2}})\int_0^\infty \int_0^t|G(x,t)|(1\wedge x^{-\frac{1}{2}})ds\mu(dx),\\
&\int_0^{\infty}\int_0^t|G(x,s)Z^x_s(\omega)|ds\nu(dx)\leq(1\vee t^{\frac{3}{2}}) \int_0^\infty\int_0^t\left|G(x,s)\right|(1\wedge x^{-\frac{3}{2}})ds\nu(dx).
\end{aligned}
\]
\end{lemma}

\begin{proof}
Note that for each $s\geq0$ the random variables $|Y^x_s|$ and $|Z^x_s|$ are half-normal distributed with mean
\[
\mathbb E\left[|Y^x_s|\right]=\sqrt{\frac{1-e^{-2sx}}{\pi x}},\quad\text{and}\quad\mathbb E\left[|Z^x_s|\right]=\sqrt{\frac{1-e^{-2sx}\left(1+2sx+2s^2x^2\right)}{2\pi x^3}}.
\]
By \eqref{equ:elem_ineq4} we have
\[
\begin{aligned}
\int_0^\infty\int_0^t\mathbb E\left[|G(x,s)Y^x_s|\right]ds\mu(dx)&=\int_0^\infty\int_0^t|G(x,s)|\sqrt{\frac{1-e^{-2sx}}{\pi x}}ds\mu(dx)\\&\leq(1\vee t^{\frac{1}{2}})\int_0^\infty\int_0^t|G(x,s)|(1\wedge x^{-\frac{1}{2}})ds\mu(dx).
\end{aligned}
\]
By \eqref{equ:elem_ineq6} we have
\[
\begin{aligned}
\int_0^\infty\int_0^t\mathbb E\left[|G(x,s)Z^x_s|\right]ds\nu(dx)&=\int_0^\infty\int_0^t|G(x,s)|\sqrt{\frac{1-e^{-2sx}\left(1+2sx+2s^2x^2\right)}{2\pi x^3}}ds\mu(dx)\\&\leq(1\vee t^{\frac{3}{2}})\int_0^\infty\int_0^t|G(x,s)|(1\wedge x^{-\frac{3}{2}})ds\nu(dx).
\end{aligned}
\]
Then the inequalities hold true with probability one.
\end{proof}

\begin{lemma}[Tightness]\label{lem:tightness}
Let $\mu_\infty,\nu_\infty$ satisfy \autoref{ass:integrability_inf}. Then the laws of the random variables $(Y_t,Z_t)_{t\geq 0}$ are tight on the space $L^1(\mu_\infty)\times L^1(\nu_\infty)$ with the weak topology. 
\end{lemma}

\begin{proof}
We generalize the proof of \cite[Proposition~2]{carmona1998fractional} to our setting. We endow $L^1(\mu_\infty)\times L^1(\nu_\infty)$ with the weak topology and assume that $(Y_0,Z_0)=0$. We will show using \cite[Theorem IV.8.9]{dunford1958linear} that for any $M\geq 0$, the set
\begin{equation*}
K_M = \left\{(y,z)\in L^1(\mu_\infty)\times L^1(\nu_\infty)\colon
\|y\|_{L^2(x^{1/2}\mu_\infty)}^2+\|z\|_{L^2(x^{1/2}\nu_\infty)}^2 \leq M\right\}
\end{equation*}
is pre-compact in $L^1(\mu_\infty)\times L^1(\nu_\infty)$. For any measurable set $E \subseteq [0,\infty)$ and $(y,z) \in K_M$, the Cauchy-Schwartz inequality implies
\begin{align*}
\|1_E y\|_{L^1(\mu_\infty)}
&\leq 
\|y\|_{L^2(x^{1/2}\mu_\infty)} \|1_E\|_{L^2(x^{-1/2}\mu_\infty)}
\leq \sqrt{M} \|1_E\|_{L^2(x^{-1/2}\mu_\infty)}, 
\\
\|1_E z\|_{L^1(\mu_\infty)}
&\leq 
\|z\|_{L^2(x^{1/2}\mu_\infty)} \|1_E\|_{L^2(x^{-1/2}\mu_\infty)}
\leq \sqrt{M} \|1_E\|_{L^2(x^{-1/2}\mu_\infty)}.
\end{align*}
Setting $E=[0,\infty)$ shows that $K_M$ is bounded in $L^1(\mu_\infty)\times L^1(\nu_\infty)$. Moreover, if $E_n \subset [0,\infty)$ is a sequence of measurable sets which decreases to the empty set, then the above estimate shows that 
\begin{equation*}
\lim_{n\to\infty} \sup_{(y,z)\in K} 
\|1_{E_n} y\|_{L^1(\mu_\infty)}+\|1_{E_n} z\|_{L^1(\mu_\infty)} = 0.
\end{equation*}
Therefore, the conditions of \cite[Theorem IV.8.9]{dunford1958linear} are satisfied and $K_M$ is pre-compact. 

By Prokhorov's theorem, the laws of $(Y_t,Z_t)_{t\geq 0}$ are tight if
\begin{equation*}
\lim_{M\to \infty} \sup_{t\geq 0} \mathbb Q[(Y_t,Z_t)\notin K_M] = 0.
\end{equation*}
This follows from the estimate
\begin{align*}
\mathbb Q[(Y_t,Z_t)\notin K_M] 
&\leq
\frac1M \mathbb E\left[\|Y_t\|_{L^2(x^{1/2}\mu_\infty)}^2
+\|Z_t\|_{L^2(x^{1/2}\nu_\infty)}^2\right]
\\&=
\frac1M \left( \int_0^\infty \operatorname{Cov}(Y^x_t) \sqrt{x}\mu_\infty(dx)
+\operatorname{Cov}(Z^x_t)\sqrt{x}\nu_\infty(dx) \right)
\\&=
\frac1M \left( \int_0^\infty \operatorname{Cov}(Y^x_\infty) \sqrt{x}\mu_\infty(dx)
+\operatorname{Cov}(Z^x_\infty)\sqrt{x}\nu_\infty(dx) \right)
\\&=
\frac1M \left(\int_0^\infty \frac{1}{2x}\sqrt{x}\mu_\infty(dx)
+ \int_0^\infty \frac{1}{4x^2}\sqrt{x}\nu_\infty(dx)\right),
\end{align*}
where the right-hand side is finite by \autoref{ass:integrability_inf}.
\end{proof}

\subsection{Auxiliary results for \autoref{sec:fractional_sr}}

\begin{lemma}[Integrability condition]\label{lem:bound_on_p}
Under \autoref{ass:integrability_shortrate}, the following condition is satisfied:
\begin{equation*}
\sup_{x \in (0,\infty)} p(x) \int_0^t se^{-sx} ds <\infty.
\end{equation*}
\end{lemma}

\begin{proof}
By assumption, there is $\beta\in(0,2)$ such that $p(x)(1\wedge x^{-\beta})$ is bounded in $x$. Therefore, the right-hand side of the following estimate is bounded in $x$, 
\begin{align*}
\int_0^t se^{-sx} ds 
&\leq
\int_0^t s\left(1\vee\left(\frac{s}{\beta}\right)^{-\beta}\right)\left(1\wedge x^{-\beta}\right)ds  
\leq
\int_0^t \left(s+\left(\frac{s}{\beta}\right)^{1-\beta}\right)ds \left(1\wedge x^{-\beta}\right)
\\&=
\left(\frac{t^2}{2}+\frac{1}{2-\beta}\left(\frac{t}{\beta}\right)^{2-\beta}\right)  \left(1\wedge x^{-\beta}\right),
\end{align*}
and the lemma follows.
\end{proof}

\begin{lemma}[Time-integrals of $(Y,Z)$]\label{lem:representation_int_inf_dim_ou}
Let \autoref{ass:integrability_shortrate} be in place and assume $(Y_0,Z_0)\in L^1(\mu)\times L^1(\nu)$ a.s. Then, for each $0\leq t\leq T$ and for all $(u,v)\in L^\infty(\mu;\mathbb C)\times L^\infty(\nu;\mathbb C)$ one has
\begin{equation*}
\int_t^T\big(\left\langle Y_s,u\right\rangle_\mu+\left\langle Z_s,v\right\rangle_\nu\big) ds
=-\left\langle Y_t,\Phi_1(T-t,u,v)\right\rangle_\mu
-\left\langle Z_t,\Phi_2(T-t,u,v)\right\rangle_\nu
-\int_t^T\left\langle\Phi_1(T-s,u,v),1\right\rangle_\mu dW_s
\end{equation*}
with $\Phi_1,\Phi_2$ as in \autoref{thm:fractional_sr_zcb}. 
In particular, the random variable $\int_t^T(\left\langle Y_s,u\right\rangle_\mu+\left\langle Z_s,v\right\rangle_\nu)ds$ is Gaussian, given $\mathcal F_t$.
\end{lemma}

\begin{proof}
The time-derivatives of $\Phi_1,\Phi_2$ are given by
\begin{align*}
\partial_\tau \Phi_1(\tau,u,v)(x)&=-e^{-\tau x}\big(u(x)+ \tau p(x)v(x)\big),&
\partial_\tau\Phi_2(\tau,u,v)(x)&=-e^{-\tau x} v(x).	
\end{align*}
It follows from \autoref{lem:representation_inf_dim_ou} that for any $0\leq t\leq s$,
\begin{equation*}
\langle Y_s,u\rangle+\langle Z_s,v\rangle
=
-\langle Y_t,\partial_\tau \Phi_1(s-t,u,v)\rangle_\mu 
- \langle Z_t,\partial_\tau \Phi_2(s-t,u,v)\rangle_\mu 
- \int_t^s \langle \partial_\tau\Phi_1(s-r,u,v),1\rangle_\mu dW_r.
\end{equation*}
The result follows by integrating over $s \in [t,T]$ and applying Fubini's theorem (\autoref{thm:fubini}) to each of the three summands above. For the first summand, Condition \eqref{equ:fubini1} of \autoref{thm:fubini} is satisfied by \autoref{lem:bound_on_p} and the estimate
\[
\begin{aligned}
&\int_0^\infty \int_t^T |Y^x_t \partial_\tau \Phi_1(s-t,u,v)| ds \mu(dx)
\\
&\qquad\leq
\|u\|_{L^\infty(\mu)} \|Y_t\|_{L^1(\mu)}
+\|v\|_{L^\infty(\nu)} \int_0^\infty |Y^x_t|\int_t^T (s-t)e^{-(s-t) x} ds p(x) \mu(dx)
\\
&\qquad=
\|u\|_{L^\infty(\mu)} \|Y_t\|_{L^1(\mu)}
+\|v\|_{L^\infty(\nu)} \int_0^\infty |Y^x_t|\int_0^{T-t} se^{-sx} ds p(x) \mu(dx)<\infty.
\end{aligned}
\]
For the second summand, Condition \eqref{equ:fubini1} reads as
\begin{equation*}\begin{aligned}
\int_0^\infty \int_t^T |Z^x_t e^{-(s-t) x} v(x)| ds\nu(dx)
\leq (T-t)\|v\|_{L^\infty(\nu)}\|Z_t\|_{L^1(\nu)}<\infty.
\end{aligned}\end{equation*}
For the third summand, we first use Fubini's theorem to exchange the order of integration with respect to $\mu(dx)$ and $dW_r$:
\begin{equation*}
\int_t^s \langle \partial_\tau\Phi_1(s-r,u,v),1\rangle_\mu dW_r
=
-\int_0^\infty\int_t^s e^{-(s-r) x}\big(u(x)+ (s-r) p(x)v(x)\big) dW_r \mu(dx).
\end{equation*}
This is allowed because \autoref{equ:fubini2} is satisfied by Equations~\eqref{equ:elem_int5} and \eqref{equ:elem_int9}:
\begin{align*}
&\int_0^\infty \sqrt{\int_t^s e^{-2(s-r)x}} |u(x)|\mu(dx)<\infty,
&
&\int_0^\infty \sqrt{\int_t^s (s-r)^2e^{-2(s-r)x}} |v(x)|\nu(dx)<\infty,
\end{align*}
Then we interchange the order of integration with respect to $dW_r$ and the product measure $\mu(dx) ds$, which brings the third summand into the form
\begin{multline*}
-\int_t^T \int_0^\infty\int_t^s e^{-(s-r) x}\big(u(x)+ (s-r) p(x)v(x)\big) dW_r \mu(dx) ds
\\=
-\int_t^T \int_r^T \int_0^\infty e^{-(s-r) x}\big(u(x)+ (s-r) p(x)v(x)\big) \mu(dx) ds dW_r. 
\end{multline*}
This is allowed because Condition \eqref{equ:fubini2} is satisfied by Equations~\eqref{equ:elem_int5_detail} and \eqref{equ:elem_int6_detail}:
\begin{align*}
&\int_t^T\int_0^\infty \sqrt{\int_t^s e^{-2(s-r)} dr}|u(x)|\mu(dx)ds<\infty,
&
&\int_t^T\int_0^\infty \sqrt{\int_t^s (s-r)^2e^{-2(s-r)} dr}|v(x)|\nu(dx)ds<\infty.
\end{align*}
Finally, we exchange the innermost integrals $\mu(dx)$ and $ds$, which is justified by Condition \eqref{equ:fubini1} and Equations~\eqref{equ:elem_int10} and \eqref{equ:elem_int11}. Then the third summand is given by
\begin{equation*}
-\int_t^T \int_0^\infty \int_r^T e^{-(s-r) x}\big(u(x)+ (s-r) p(x)v(x)\big)  ds \mu(dx) dW_r
=
-\int_t^T\left\langle\Phi_1(T-r,u,v),1\right\rangle_\mu dW_r.\qedhere
\end{equation*}
\end{proof}

\begin{lemma}[Semimartingale property]\label{lem:semimartingale_Psi}
Under \autoref{ass:integrability_shortrate}, the expressions $\left\langle Y_t,\Phi_1(T-t,u,v)\right\rangle_\mu$ and $\left\langle Z_t,\Phi_2(T-t,u,v)\right\rangle_\nu$
are continuous semimartingales in $t \in [0,T]$, for each fixed $T>0$ and $(u,v)\in L^\infty(\mu)\times L^\infty(\nu)$.
\end{lemma}

\begin{proof}
We verify the conditions of \autoref{thm:semimartingale}. In the following estimates it can be assumed without loss of generality that the functions $u$ and $v$ are equal to 1 because they are bounded. Conditions \eqref{equ:semimartingale1} and \eqref{equ:semimartingale2} for $f^x_t=\Phi_1(T-t,u,v)(x)$ are satisfied by Equations~\eqref{equ:elem_int12},  \eqref{equ:elem_int13} and \eqref{equ:elem_int14}:
\begin{align*}
\int_0^\infty \int_0^t|\partial_s f^x_s-xf^x_s|(1\wedge x^{-\frac12})ds\mu(dx)
&=\int_0^\infty\int_0^t \left(1+\frac{1-e^{-(T-s) x}}{x} p(x)\right)(1\wedge x^{-\frac12})ds\mu(dx)<\infty,
\\
\int_0^\infty \sqrt{\int_0^t (f^x_s)^2 ds}\mu(dx)
&\leq
\int_0^\infty \sqrt{\int_0^t 2\left(\frac{e^{-(T-s) x}-1}{x}\right)^2 ds}\mu(dx) 
\\&\quad
+\int_0^\infty \sqrt{\int_0^t 2\left(\frac{e^{-(T-s) x}-1}{x^2}+\frac{\tau}{x}e^{-(T-s) x}\right)^2 ds}\nu(dx) < \infty.
\end{align*}
Conditions \eqref{equ:semimartingale3} and \eqref{equ:semimartingale4} are satisfied for $g^x_t=\Phi_2(T-t,u,v)(x)$ by \autoref{equ:elem_int12}:
\begin{align*}
&\int_0^\infty \int_0^t|\partial_s g^x_s-xg^x_s|(1\wedge x^{-\frac32})ds\nu(dx)
=\int_0^\infty \int_0^t (1\wedge x^{-\frac32})ds\nu(dx)<\infty,\\ 
&\int_0^\infty \int_0^t|g^x_s|(1\wedge x^{-\frac12})ds\nu(dx)
=\int_0^\infty \int_0^t\frac{1-e^{-\tau x}}{x}(1\wedge x^{-\frac12})ds\nu(dx)<\infty.
\end{align*}
Thus, we have verified the conditions of \autoref{thm:semimartingale} and the statement of the lemma follows.
\end{proof}

\begin{lemma}[Semimartingale property]\label{lem:semimartingale_Psi_prime}
Under \autoref{ass:integrability_shortrate}, the expressions $\langle Y_t,\partial_\tau\Phi_1(\tau,u,v)\rangle_\mu$ and $\langle Z_t,\partial_\tau\Phi_2(\tau,u,v)\rangle_\nu$ are continuous semimartingales in $t\in [0,T]$, for each fixed $\tau>0$ and $(u,v) \in L^\infty(\mu)\times L^\infty(\nu)$.
\end{lemma}

\begin{proof}
We show the semimartingale property by verifying the conditions of \autoref{thm:semimartingale}. We have
\begin{align*}
\partial_\tau\Phi_1(\tau,u,v)(x)&=-e^{-\tau x}\big(u(x)+\tau p(x)v(x)\big),
&
\partial_\tau\Phi_2(\tau,u,v)(x)&=-e^{-\tau x} v(x).
\end{align*}
In the following estimates it can be assumed without loss of generality that the functions $u$ and $v$ are equal to 1 because they are bounded. Conditions \eqref{equ:semimartingale1} and \eqref{equ:semimartingale2} for $f^x_t=\partial_\tau\Phi_1(\tau,u,v)(x)$ are satisfied by Equations~\eqref{equ:elem_int1}--\eqref{equ:elem_int4}:
\begin{align*}
\int_0^\infty \int_0^t|\partial_s f^x_s-xf^x_s|(1\wedge x^{-\frac12})ds\mu(dx)
&=\int_0^\infty\int_0^t xe^{-\tau x} \big(1+\tau p(x)\big)(1\wedge x^{-\frac12})ds\mu(dx)<\infty,
\\
\int_0^\infty \sqrt{\int_0^t (f^x_s)^2 ds}\mu(dx)
&\leq
\int_0^\infty \sqrt{\int_0^t 2e^{-2\tau x} ds}\mu(dx) 
+\int_0^\infty \sqrt{\int_0^t 2\tau^2 e^{-2\tau x} ds}\nu(dx) < \infty.
\end{align*}
Conditions \eqref{equ:semimartingale3} and \eqref{equ:semimartingale4} for $g^x_t=\partial_\tau\Phi_2(\tau,u,v)(x)$ are satisfied by \autoref{equ:elem_int4}: 
\begin{align*}
&\int_0^\infty \int_0^t|\partial_s g^x_s-xg^x_s|(1\wedge x^{-\frac32})ds\nu(dx)
=\int_0^\infty \int_0^t xe^{-\tau x} (1\wedge x^{-\frac32})ds\nu(dx)<\infty,\\ 
&\int_0^\infty \int_0^t|g^x_s|(1\wedge x^{-\frac12})ds\nu(dx)
=\int_0^\infty \int_0^txe^{-\tau x} (1\wedge x^{-\frac12})ds\nu(dx)<\infty.
\end{align*}
Thus, we have verified the conditions of \autoref{thm:semimartingale} and the statement of the lemma follows.
\end{proof}

\subsection{Auxiliary results for \autoref{sec:fractional_ba}}

\begin{lemma}[Semimartingale property]\label{semimartingale_psi_prime}
Under \autoref{ass:integrability1}, the expressions $\left\langle Y_t,\partial_\tau\phi_1(\tau,-u,-v)\right\rangle_\mu$ and $\left\langle Z_t,\partial_\tau\phi_2(\tau,-u,-v)\right\rangle_\nu$ are continuous semimartingales in $t \in [0,\infty)$ for each fixed $\tau>0$ and $(u,v) \in L^\infty(\mu)\times L^\infty(\nu)$.
\end{lemma}

\begin{proof}
We verify the conditions of \autoref{thm:semimartingale}. As $u$ and $v$ are bounded we may assume without loss of generality in the following estimates that $u=v=1$. Conditions \eqref{equ:semimartingale1}--\eqref{equ:semimartingale4} for $f^x_t=\partial_\tau\phi_1(\tau,-u,-v)(x)$ and $g^x_t=\partial_\tau\phi_2(\tau,-u,-v)(x)$ are satisfied by Equations \eqref{equ:elem_int1}--\eqref{equ:elem_int4}:
\begin{align*}
&\int_0^\infty \int_0^t|\partial_s f^x_s-xf^x_s|(1\wedge x^{-\frac12})ds\mu(dx)
=t\int_0^\infty\left|x\partial_\tau\phi_1(\tau,-u,-v)\right|(1\wedge x^{-\frac12})\mu(dx)
\\&\qquad\leq t\int_0^\infty x^2e^{-x\tau}\mu(dx)+t\int_0^\infty xe^{-x\tau}\nu(dx)+t\tau\int_0^\infty x^2e^{-x\tau}\nu(dx)<\infty,
\\
&\int_0^\infty \sqrt{\int_0^t (f^x_s)^2 ds}\mu(dx)
=\sqrt{t}\int_0^\infty\left|\partial_\tau\phi_1(\tau,-u,-v)\right|\mu(dx)
\\&\qquad\leq\sqrt{t}\int_0^\infty xe^{-x\tau}\mu(dx)+\sqrt{t}\int_0^te^{-x\tau}\nu(dx)+\sqrt{t}\tau\int_0^txe^{-x\tau}\nu(dx)<\infty,
\\
&\int_0^\infty \int_0^t|\partial_s g^x_s-xg^x_s|(1\wedge x^{-\frac32})ds\nu(dx)
=t\int_0^\infty\left|x\partial_\tau\phi_2(\tau,-u,-v)\right|(1\wedge x^{-\frac32})\nu(dx)
\\&\qquad
=t\int_0^\infty x^2e^{-x\tau}\nu(dx)<\infty,
\\ 
&\int_0^\infty \int_0^t|g^x_s|(1\wedge x^{-\frac12})ds\nu(dx)
=t\int_0^\infty|\partial_\tau\phi_2(\tau,-u,-v)|\nu(dx)
\leq t\int_0^{\infty}e^{-x\tau}\nu(dx)<\infty.
\qedhere
\end{align*}
\end{proof}

\subsection{Auxiliary results for \autoref{sec:stein}}

\begin{lemma}[Injectivity of the covariance operator]\label{lem:injectivity}
For any $\tau>0$, the mapping $P_\tau$ is an injective linear operator from $L^\infty(\mu;\mathbb C)$ to the complexification of the Hilbert space $\mathsf H_\tau$.
\end{lemma}

\begin{proof}
For simplicity, we write $\mathsf H_\tau$ for the complexified space $\mathsf H_\tau \otimes_{\mathbb R} \mathbb C$ (see \autoref{sec:reproducing}). If $P_\tau v =0$ for some $v \in L^\infty(\mu;\mathbb C)$, then 
\[
0=\langle P_\tau v,P_\tau v\rangle_{\mathsf H_\tau}= \langle P_\tau v,v\rangle_\mu=\int_0^\tau\left(\int_0^\infty v(x) e^{-sx}\mu(dx)\right)ds.
\]
Therefore, the Laplace transform $\mathcal L(v\mu)(s)$ of the complex measure $v\mu$ vanishes at almost all $s \in [0,\tau]$. As $\mathcal L(v\mu)(s)$ is analytic in $s$, it vanishes identically. By the injectivity of the Laplace transform \cite[Section~3.8]{jacob2001pseudo}, the complex measure $v\mu$ vanishes, which is equivalent to $v=0$ in $L^\infty(\mu;\mathbb C)$.
\end{proof}

\begin{lemma}[Diagonalization of symmetric two-tensors]\label{lem:diagonalization}
For each $\tau\geq0$ any symmetric two-tensor $w \in L^\infty(\mu;\mathbb C)^{\otimes2}$ has a representation as a sum of squares
\[
w = \sum_{k=1}^n \vartheta_k v_k \otimes v_k,\quad\text{with $\vartheta_k \in\mathbb C$ and $v_k \in L^\infty(\mu;\mathbb C)$},
\]
such that the functions $v_k$ are orthornormal with respect to the covariance operator $P_\tau$ defined in \autoref{lem:cov}, i.e., $\langle P_\tau v_k,v_l\rangle_\mu=\delta_{kl}$.
\end{lemma}

\begin{proof}
For simplicity, we write $\mathsf H_\tau$ for the complexified space $\mathsf H_\tau \otimes_{\mathbb R} \mathbb C$. Let $w=\sum_{k=1}^m w_k\otimes w_k \in L^\infty(\mu;\mathbb C)^{\otimes2}$ be any symmetric two-tensor and set $V=\operatorname{span}_{\mathbb C} \{w_1,\dots,w_m\}$. By \autoref{lem:injectivity}, the bilinear form $\langle P_\tau\cdot,\cdot\rangle$ is a scalar product on the finite-dimensional vector space $V$. The desired representation of $w$ is obtained by diagonalizing $w \in V^{\otimes 2}$ with respect to this scalar product.
\end{proof}

\begin{lemma}[Affine structure]\label{lem:stein:affine_general}
Let $\mu$ satisfy \autoref{ass:integrability1} and $Y_0\in L^1(\mu)$. Let $w=\sum_{k=1}^n \vartheta_k v_k^{\otimes2}\in iL^\infty(\mu)^{\otimes2}$ be a symmetric tensor with decomposition into sums of squares in the sense of \autoref{lem:diagonalization}, and $0\leq t\leq T$,
\[
\mathbb E\left[e^{\left\langle\Pi_T,w\right\rangle_{\mu^{\otimes2}}}\middle|\mathcal F_t\right]=e^{\psi_0(T-t,w)+\left\langle\Pi_t,\psi_1(T-t,w)\right\rangle_{\mu^{\otimes2}}},
\]
where $(\psi_0,\psi_1)\colon[0,\infty)\times L^\infty(\mu;\mathbb C)^{\otimes2}\to\mathbb C \times L^\infty(\mu;\mathbb C)^{\otimes2}$ are given by
\begin{align*}
\psi_0(\tau,w)&=-\frac{1}{2}\sum_{k=1}^{n}\log\left(1-2\vartheta_k\right),
&
\psi_1(\tau,w)(x,y)&=\sum_{k=1}^{n}\frac{\vartheta_k}{1-2\vartheta_k}v_k(x)v_k(y)e^{-(T-t)(x+y)}.
\end{align*}
\end{lemma}

\begin{proof}
Let $0\leq t\leq T$ be fixed and let $w=\sum_{k=1}^n \vartheta_k v_k^{\otimes2}$ be a decomposition of $w$ into sums of squares in the sense of \autoref{lem:diagonalization}. By \autoref{lem:representation_inf_dim_ou,lem:cov,lem:diagonalization}  the random variables $\left\langle Y_T,v_1\right\rangle_\mu,\ldots,\left\langle Y_T,v_n\right\rangle_\mu$ are independent Gaussian, given $\mathcal F_t$, with conditional means
\[
\mathbb E\left[\left\langle Y_T,v_k\right\rangle_\mu\middle|\mathcal F_t\right]=\left\langle Y_t,\phi_1(T-t,v_k,0)\right\rangle_\mu,\quad k\in\{1,\ldots,n\},
\]
and unit variances. Hence, the random variables $\left\langle Y,v_1\right\rangle^2_\mu,\ldots,\left\langle Y,v_n\right\rangle^2_\mu$ are independent non-central $\chi^2$, given $\mathcal F_t$, with non centrality parameters
\[
\left\langle Y_t,\phi_1(T-t,v_k,0)\right\rangle^2_\mu=\left\langle \Pi_t,\phi_1(T-t,v_k,0)^{\otimes2}\right\rangle_{\mu^{\otimes2}},\quad k\in\{1,\ldots,n\}.
\]
We obtain the affine transformation formula using independence and the characteristic function of the non-central $\chi^2$ distribution
\begin{align*}
\mathbb E\left[e^{\left\langle\Pi_T,w\right\rangle_{\mu^{\otimes2}}}\middle|\mathcal F_t\right]&=\prod_{k=1}^{n}\mathbb E\left[e^{\vartheta_k\langle Y_T,v_k\rangle^2_\mu}\middle|\mathcal F_t\right]
\\&=\exp\left(-\frac{1}{2}\sum_{k=1}^{n}\log\left(1-2\vartheta_k\right)+\sum_{k=1}^{n}\frac{\vartheta_k}{1-2\vartheta_k}\left\langle\Pi_t,\phi_1(T-t,v_k,0)^{\otimes2}\right\rangle_{\mu^{\otimes2}}\right).
\end{align*}
We recognize the functions $\psi_0$ and $\psi_1$ on the right-hand side above.
\end{proof}

\begin{lemma}[Conditional mean]\label{lem:stein:mean}
For each $v^{\otimes2} \in L^\infty(\mu)\otimes_s L^\infty(\mu)$ and $0\leq t\leq T$, the $\mathcal F_t$-conditional mean of $\left\langle\Pi_T,v^{\otimes2}\right\rangle_{\mu^{\otimes2}}$ is given by
\[
\begin{aligned}
\mathbb E\left[\left\langle\Pi_T,v^{\otimes2}\right\rangle_{\mu^{\otimes2}}\middle|\mathcal F_t\right]
&=
2\phi_0(\tau,v,0)+\left\langle\Pi_t,\phi_1(\tau,v,0)^{\otimes2}\right\rangle_{\mu^{\otimes2}}.
\end{aligned}
\]
\end{lemma}

\begin{proof}
We use \autoref{thm:stein:affine:special} to calculate
\begin{align*}
&\mathbb E\left[\left\langle\Pi_T,v^{\otimes2}\right\rangle_{\mu^{\otimes2}}\middle|\mathcal F_t\right]
=\frac{1}{i}\partial_q|_{q=0}\mathbb E\left[e^{iq\left\langle\Pi_T,v^{\otimes2}\right\rangle_{\mu^{\otimes2}}}\middle|\mathcal F_t\right]
=\frac{1}{i}\partial_q|_{q=0} e^{\psi_0\left(T-t,v^{\otimes2}iq\right)+\left\langle\Pi_t,\psi_1\left(T-t,v^{\otimes2}iq\right)\right\rangle_{\mu^{\otimes2}}}
\\&\quad=
\frac{1}{i}\partial_q|_{q=0} e^{-\frac{1}{2}\log\left(1-4\phi_0(\tau,v\sqrt{iq},0)\right)+\left\langle\Pi_t,\frac{\phi_1(\tau,v\sqrt{iq},0)^{\otimes2}}{1-4\phi_0(\tau,v\sqrt{iq},0)}\right\rangle_{\mu^{\otimes2}}}
\\&\quad=
\frac{1}{i}\partial_q|_{q=0} e^{-\frac{1}{2}\log\left(1-4iq\phi_0(\tau,v,0)\right)+iq\left\langle\Pi_t,\frac{\phi_1(\tau,v,0)^{\otimes2}}{1-4iq\phi_0(\tau,v,0)}\right\rangle_{\mu^{\otimes2}}}
=2\phi_0(\tau,v,0)+\left\langle\Pi_t,\phi_1(\tau,v,0)^{\otimes2}\right\rangle_{\mu^{\otimes2}}.
\qedhere
\end{align*}
\end{proof}

\begin{lemma}[Conditional second moment]\label{lem:stein:covariance}
Let $w \in L^\infty(\mu)\otimes L^\infty(\mu)$ be a symmetric tensor with sum-of-squares representation
\[
w=\sum_{k=1}^n \vartheta_kv_k^{\otimes2}.
\]
as in \autoref{lem:diagonalization}. Then for each $0\leq t\leq T$, the $\mathcal F_t$-conditional second moment of $\left\langle\Pi_T,w\right\rangle_{\mu^{\otimes2}}$ is given by
\[
\begin{aligned}
\mathbb E\left[\left\langle \Pi_T,w\right\rangle^2_{\mu^{\otimes2}}\middle|\mathcal F_t\right]&=2\sum_{k=1}^{n}\vartheta^2_k+2\sum_{k=1}^{n}\vartheta^2_k\left\langle\Pi_t,\phi_1(T-t,v_k,0)^{\otimes2}\right\rangle_{\mu^{\otimes2}}\\&\quad+\left(\sum_{k=1}^n\vartheta_k+\sum_{k=1}^n\vartheta_k\left\langle\Pi_t,\phi_1(T-t,v_k,0)^{\otimes2}\right\rangle_{\mu^{\otimes2}}\right)^2.
\end{aligned}
\]
\end{lemma}

\begin{proof}
As in the proof of \autoref{lem:stein:affine_general}, we have
\[
\begin{aligned}
\psi_0(T-t,iqw,0)&=-\frac{1}{2}\sum_{k=1}^n\log\left(1-2iq\vartheta_k\right),\\
\psi_1(T-t,iqw,0)(x,y)&=e^{-(T-t)(x+y)}\sum_{k=1}^n\frac{iq\vartheta_k}{1-2iq\vartheta_k}v_k(x)v_k(y).
\end{aligned}
\]
For the derivatives of $\psi_0(T-t,iqw,0)$ with respect to $q$ we have
\begin{align*}
\partial_q\psi_0(T-t,iqw,0)&=i\sum_{k=1}^n\frac{\vartheta_k}{1-2iq\vartheta_k},
&
\partial^2_q\psi_0(T-t,iqw,0)&=-2\sum_{k=1}^n\frac{\vartheta^2_k}{\left(1-2iq\vartheta_k\right)^2},
\end{align*}
and for the derivatives of $\psi_1(T-t,iqw,0)$ with respect to $q$ we have
\[
\begin{aligned}
\partial_q\psi_1(T-t,iqw,0)(x,y)&=ie^{-(T-t)(x+y)}\sum_{k=1}^n\frac{\vartheta_k}{\left(1-2iq\vartheta_k\right)}v_k(x)v_k(y),\\
\partial^2_q\psi_1(T-t,iqw,0)(x,y)&=-2e^{-(T-t)(x+y)}\sum_{k=1}^n\frac{\vartheta^2_k}{\left(1-2iq\vartheta_k\right)^2}v_k(x)v_k(y).
\end{aligned}
\]
Using the characteristic function we obtain
\begin{align*}
\mathbb E\left[\left\langle \Pi_T,w\right\rangle^2_{\mu^{\otimes2}}\middle|\mathcal F_t\right]&=-\partial^2_q|_{q=0}\mathbb E\left[e^{iq\left\langle\Pi_T,w\right\rangle_{\mu^{\otimes2}}}\middle|\mathcal F_t\right]\\&=2\sum_{k=1}^{n}\vartheta^2_k+2\sum_{k=1}^{n}\vartheta^2_k\left\langle\Pi_t,\phi_1(T-t,v_k,0)^{\otimes2}\right\rangle_{\mu^{\otimes2}}
\\&\quad+\left(\sum_{k=1}^n\vartheta_k+\sum_{k=1}^n\vartheta_k\left\langle\Pi_t,\phi_1(T-t,v_k,0)^{\otimes2}\right\rangle_{\mu^{\otimes2}}\right)^2.\qedhere
\end{align*}
\end{proof}

\printbibliography

\end{document}